%% file: main.tex
\documentclass[runningheads]{llncs}

\usepackage{doc}

\usepackage{amsmath}

\usepackage{amssymb}
\usepackage{xcolor}
\usepackage{algorithm}
\usepackage{algpseudocode}
\usepackage{caption}
\usepackage{subcaption}

\usepackage{ebproof}

\usepackage{comment}
\usepackage{tikz}
\usetikzlibrary{calc}
\usepackage{standalone}
\usepackage{import}
\usepackage{booktabs}
\usepackage{array}
\usepackage{makecell}
\usepackage{multirow}
\usepackage{wrapfig}  

\usepackage{tabularx}

\usepackage{mathtools}

\BeforeBeginEnvironment{wrapfigure}{\setlength{\intextsep}{0pt}}

\usepackage{tikz}
\usetikzlibrary{positioning,arrows,shapes,hobby,backgrounds,calc,trees,fit}
\usetikzlibrary{arrows,automata,shadows,shapes.arrows}
\usepackage{tikz-cd}
\usepackage[inkscapelatex=false]{svg}
\usepackage{enumitem}
\usepackage{mdframed}

\input{envs}

\input{macro}

\title{Specializing anti-unification for interaction models composition via gate connections}
\titlerunning{Constrained anti-unification: application to interaction models composition}
\author{
\!\!\!\!\!\!Joel Nguetoum\inst{1}\inst{2}
\!\and\!\!
Boutheina Bannour \inst{1}
\!\!\!\and\!\!
Pascale Le Gall\inst{2}
\!\!\!\and\!\!
Erwan Mahe \inst{1}
}

\authorrunning{ J. Nguetoum, B. Bannour,and P. Le Gall}
\institute{
Université Paris-Saclay, CEA, List, F-91120, Palaiseau, France
\and
Université Paris-Saclay, CentraleSupélec, MICS, F-91192, Gif-sur-Yvette, France
}
\AddToHook{cmd/appendix/before}{}
\begin{document}

\maketitle

\begin{abstract}
Interaction models describe distributed systems as algebraic terms, with gates marking interaction points between local views. Composing local models into a coherent global one requires aligning these gates while respecting the algebraic laws of interaction operators. We specialize anti-unification (or generalization) via a special constant-preserving 
variant, which preserves designated constants while generalizing the remaining structure. We develop a dedicated rule-based procedure, for computing these generalizations, 
prove its termination, soundness, and completeness, extend it modulo equational theories, and integrate it into a standard anti-unification framework. A prototype tool demonstrates the approach’s ability to recompose global interactions from partial views.

\end{abstract}

\input{input/1_intro}

\input{input/2_0_prelim}

\input{input/3_0_scp_genrlz}

\input{input/4_0_algo_sansE}

\input{input/5_0_equational_scp_genrlz}

\input{input/7_0_composition}

\input{input/9_conclusion}

\bibliographystyle{splncs04}
\bibliography{biblio/biblio}

\newpage
\appendix

\input{appendix/Additional_examples}

\end{document}

%% file: envs.tex
\spnewtheorem*{lemma*}{Lemma}{\bfseries}{\itshape}
\spnewtheorem*{property*}{Property}{\bfseries}{\itshape}
\spnewtheorem*{theorem*}{Theorem}{\bfseries}{\itshape}

\spnewtheorem*{sketchproof*}{Sketch of proof}{\itshape}{}

%% file: macro.tex
\newcommand{\dom}{\textsf{Dom}}
\newcommand{\head}{\textsf{head}}

\newcommand{\seq}{\textsf{seq}}
\newcommand{\vp}{\textsf{vp}}
\newcommand{\alt}{\textsf{alt}}
\newcommand{\para}{\textsf{par}}

\newcommand{\loopS}{\textsf{loop}}

\newcommand{\var}{\textsf{Var}}
\newcommand{\ran}{\textsf{Ran}}
\newcommand{\vran}{\textsf{VRan}}
\newcommand{\idx}{\textsf{Index}}

\newcommand{\size}{\textsf{size}}

\newcommand{\sleft}{\textsf{Left}}
\newcommand{\sright}{\textsf{Right}}
\newcommand{\sigmal}{\sigma_{\textsf{L}}}
\newcommand{\sigmar}{\sigma_{\textsf{R}}}

\newcommand{\termalg}{\T(\F,\V)}
\newcommand{\termground}{\T(\F)}

\newcommand{\T}{\mathcal{T}}
\newcommand{\F}{\mathcal{F}}

\newcommand{\V}{\mathcal{V}}

\newcommand{\M}{\mathcal{M}}

\newcommand{\mappingNotation}{\lambda}

\newcommand{\specialConstantPreservingGeneralisation}{scpg}
\newcommand{\specialConstantPreservingLgg}{scplgg}
\newcommand{\specialConstantPreservingMgg}{scpmgg}

\newcommand{\setOfAllSpecialConstants}{\mathfrak{F}_0}
\newcommand{\setOfAllAUTs}{\mathcal{A}}
\newcommand{\activeAUTs}{A}

\newcommand{\pos}{\textsf{pos}}

\newcommand{\leftmultiset}{\{\!\!\{}
\newcommand{\rightmultiset}{\}\!\!\}}

\newcommand{\conf}{c}

\definecolor{lightred}{rgb}{2.16, 0.38, 0.38}
\definecolor{lightpurple}{rgb}{1.57, 0.38, 2.16}
\definecolor{darkpurple}{rgb}{1.28, 0, 1.28}

\definecolor{brightmaroon}{rgb}{0.76, 0.13, 0.28}
\definecolor{brandeisblue}{rgb}{0.0, 0.44, 1.0}
\definecolor{darkpastelgreen}{rgb}{0.01, 0.75, 0.24}

\definecolor{hibou_col_lf}{RGB}{22, 22, 130}
\newcommand{\hlf}[1]{\textcolor{hibou_col_lf}{\text{#1}}}
\definecolor{hibou_col_ms}{RGB}{15, 86, 15}
\newcommand{\hms}[1]{\textcolor{hibou_col_ms}{\text{#1}}}

\definecolor{myRed}{RGB}{255,0,0}       
\definecolor{myGreen}{RGB}{0,128,0}     
\definecolor{myPurple}{RGB}{128,0,128}  
\definecolor{myYellow}{RGB}{255,204,0}  

\newcommand{\messageOfAction}{\psi}

\newcommand{\lfs}{\theta}

\newcommand{\inter}{\mathbb{I}}

\newcommand{\act}{\mathbb{A}}
\newcommand{\vps}{\mathbb{VP}}

\newcommand{\greencheck}{{\color{green}\checkmark}}

%% file: input/1_intro.tex
\section{Introduction}

\paragraph{Context.}
Our goal is to use anti-unification to merge partial views of a distributed system into a coherent global view.
The notion of \emph{anti-unification} (or \emph{generalization}) was introduced independently by Plotkin~\cite{Plotkin1970} and Reynolds~\cite{Reynolds1970} as the process of computing the most specific generalization of a set of terms. Given two terms $s$ and $t$, the task is to find a term $r$ and substitutions $\sigma_s$ and $\sigma_t$ such that $r\sigma_s = s$ and $r\sigma_t = t$. This contrasts with \emph{unification}~\cite{Unification_revisited}, which seeks a substitution $\sigma$ such that $s\sigma = t\sigma$. While unification computes a common instance, anti-unification computes a \emph{most specific/least general generalization} (lgg) from which the original terms can be recovered. Anti-unification has applications in recursion-scheme detection~\cite{appli_anti_unfication_rec_scheme_detection}, inductive synthesis of recursive functions~\cite{appli_anti_unfication_func_program_synthesis}, and automated software repair~\cite{appli_anti_unfication_bug_fixes_1,appli_anti_unfication_bug_fixes_2}. However, most existing work focuses on \emph{syntactic} generalization.
\emph{Semantics-driven} approaches \cite{CernaUnital,ALPUENTE2014} leverage semantic equivalences between terms to be able to anti-unify a wider class of terms.

Interactions~\cite{MaheSEMANTICS} describe distributed systems via atomic \emph{emission} and \emph{reception} actions and the empty behavior $\varnothing$ which are structured by \emph{weak-sequencing} (\textsf{seq}), \emph{alternative} (\textsf{alt}) and \emph{parallel} (\textsf{par}) composition, as well as \emph{repetition} (\textsf{loop}).
Interactions are related to \emph{Multiparty Session Types} \cite{twoBuyerProtocolRef} and can be drawn in the manner of \emph{Message Sequence Charts} \cite{msc} or \emph{UML sequence diagrams} \cite{uml}. The diagram on the right of Fig.~\ref{fig:example_interaction_composition} depicts an interaction $k$. Operators are represented by boxes, with \textsf{seq} being implicit and time flowing from top to bottom. Vertical lines, called \emph{lifelines}, represent individual subsystems (here, a signaling server $\hlf{ss}$, a dialing client $\hlf{dc}$, and a target client $\hlf{tc}$). Horizontal arrows represent value passing (\textsf{vp}). Here, $\hlf{dc}$ dials $\hlf{ss}$ (via message $\hms{dia}$), which then notifies $\hlf{tc}$ ($\hms{not}$). Finally, $\hlf{tc}$ can either (represented by \textsf{alt}) refuse the connection and warn another system ($\hms{wrn}$) not shown here, or accept it by sending $\hms{con}$ to $\hlf{dc}$, which is answered by an acknowledgment $\hms{ack}$.

\paragraph{Method.}
 
Our objective is to build a global interaction $k$ from partial views $i$ and $j$ defined over a partition of the lifelines (on Fig.\ref{fig:example_interaction_composition}, $k$ is defined over $\{\hlf{dc},\hlf{ss},\hlf{tc}\}$, $i$ over $\{\hlf{tc}\}$ and $j$ over $\{\hlf{dc},\hlf{ss}\}$) .
 
We represent interactions as terms, which allows us to encode their composition as an anti-unification problem (on Fig.\ref{fig:example_interaction_composition}, anti-unification allows e.g., the superimposition of the \textsf{alt} operator in $i$ and $j$).

Our originality lies in identifying the communication channels connecting these partial views as special constants called \emph{gates}. 
For example, on Fig.\ref{fig:example_interaction_composition}, the gate $\mathfrak{\textcolor{red}{a}}$ relate the emission $\hlf{tc}!\hms{not}$ in $i$ and the reception $\hlf{ss}?\hms{not}$ in $j$. These two actions must be coupled in the resulting interaction $k$ to form a value passing $\textsf{vp}(\hlf{tc}!\hms{not},\hlf{ss}?\hms{not})$.

Intermediate terms $s$ and $t$ are derived from $i$ and $j$ by substituting related actions by the gate that relates them ($\mathfrak{\textcolor{red}{a}}$, $\mathfrak{\textcolor{violet}{b}}$, and $\mathfrak{\textcolor{cyan}{c}}$ on Fig.\ref{fig:example_interaction_composition}).
This yields mappings $\mappingNotation_i$ and $\mappingNotation_j$ such that $s\mappingNotation_i = i$ and $t \mappingNotation_j = j$.

Applying our \emph{special constant-preserving anti-unification} modulo an equational theory $E$ yields a 
generalization $r$ of $s$ and $t$ 
with substitutions $\sigma_s, \sigma_t$ such that $r\sigma_s =_E s$, $r\sigma_t =_E t$. As the generalization preserves the gates, $r$ contains the three gates $\mathfrak{\textcolor{red}{a}}$, $\mathfrak{\textcolor{violet}{b}}$, and $\mathfrak{\textcolor{cyan}{c}}$. 
Here, $E$ includes $\textsf{alt}(x,y) \!\approx\! \textsf{alt}(y,x)$ (commutativity of \textsf{alt}) and $\textsf{seq}(\varnothing,x)\! \approx\! x$ ($\varnothing$ neutral element for \textsf{seq}).
The composed interaction $k$ can then be built from $r$, as illustrated on Fig.\ref{fig:example_interaction_composition}, by applying a substitution $\sigma_r$ derived from $\sigma_s$ and $\sigma_t$, along with a mapping $\mappingNotation_k$ derived from $\mappingNotation_i$ and $\mappingNotation_j$.

\begin{figure}[t]

\input{figure/example_int/introduction_example_small}

  \caption{Composing interactions $i$ and $j$ into $k$.}
  \label{fig:example_interaction_composition}
\end{figure}


\paragraph{Related work}  
Compositionality for distributed systems is an active topic in Multi-party Session Types (MPST)~\cite{LangeSynthesis,CompositionalChoreographies,BarbaneraMPSTComposition1,BarbaneraMPSTComposition2,Stolze2023}. Choreography-based languages, including MPST and our interaction language, represent systems as algebraic terms with equational laws defining operator semantics. MPSTs specify a global type and, via projection, derive local types that ensure safety and avoid deadlocks, at the cost of restrictive structural forms (e.g., “choice at sender”). \cite{LangeSynthesis} composes partial views expressed as CCS-like local types, while~\cite{CompositionalChoreographies,Stolze2023} consider partial MPSTs interacting with the environment, but remain constrained by operational assumptions. \cite{BarbaneraMPSTComposition1,BarbaneraMPSTComposition2} introduce gateways relaying messages between sessions, enabling open-system composition. Similar ideas exist for Communicating Finite State Machines (CFSMs)~\cite{BarbaneraCFSMComposition1,BarbaneraCFSMComposition2,LangeCFSMtoGlobal}, where composition relies on gateways or safe CFSM classes. Our anti-unification-based composition relies only on algebraic equations to merge partial terms, enabling a broader class of behaviors without structural or operational constraints.

\paragraph{Contributions and paper organization.}

In Sec.\ref{sec:Preliminaries} we present standard concepts and notations related to term algebras and anti-unification. In Sec.\ref{sect:constant-preserving-generalization} we introduce \emph{special constant-preserving (sc-preserving) anti-unification}, which preserves designated constants while generalizing other subterms. We establish its relation to standard anti-unification, and characterize conditions for the existence and uniqueness of sc-preserving least general generalizations (lggs).

In Sec.\ref{sec:algo_without_eqs}, we prove termination, soundness, and completeness of our \emph{rule-based algorithm} for computing sc-preserving lggs. 

In Sec.\ref{sec:extension_modulo_E}, we extend sc-preserving anti-unification \emph{modulo equational theories}~\cite{ALPUENTE2014,ALPUENTE2022}. 
In Sec.\ref{sec:application_to_int_compo}, we show that \emph{interaction composition} reduces to sc-preserving anti-unification by encoding gates as special constants to align local views, proving preservation of local behaviors via projection. 
Finally, we present an experimental validation of a prototype tool~\cite{generalizer_repo_2025} that implements sc-preserving generalization modulo equations, demonstrating its ability to recombine interactions taken from the literature.

%% file: figure/example_int/introduction_example_small.tex
\begin{tabular}{|l|l|l|l|}
\hline 
\makecell[l]{

\input{figure/example_int/i_diag}

\\
\scalebox{.9}{
$
\begin{array}{l}
u = {\scriptstyle\hlf{tc}!\hms{wrn}}\\
\mappingNotation_i =
\left\{\!
\begin{array}{l}
\mathfrak{\textcolor{red}{a}} \mapsto {\scriptstyle\hlf{tc}?\hms{not}},\\
\mathfrak{\textcolor{violet}{b}} \mapsto {\scriptstyle\hlf{tc}!\hms{con}},\\
\mathfrak{\textcolor{cyan}{c}} \mapsto {\scriptstyle\hlf{tc}?\hms{ack}}\\
\end{array}
\!\right\}
\end{array}
$
}
\\
\scalebox{.9}{
\input{figure/example_int/i_term}
}
}
&
\makecell[l]{
\input{figure/example_int/j_diag}

\\
\scalebox{.9}{
$
\begin{array}{l}
v = {\scriptstyle\textsf{vp}(\hlf{dc}!\hms{dia},\hlf{ss}?\hms{dia})}\\
\mappingNotation_j =
\left\{\!
\begin{array}{l}
\mathfrak{\textcolor{red}{a}} \mapsto {\scriptstyle\hlf{ss}!\hms{not}},\\
\mathfrak{\textcolor{violet}{b}} \mapsto {\scriptstyle\hlf{dc}?\hms{con}},\\
\mathfrak{\textcolor{cyan}{c}} \mapsto {\scriptstyle\hlf{dc}!\hms{ack}}\\
\end{array}
\!\right\}
\end{array}
$
}
\\
\scalebox{.9}{
\input{figure/example_int/j_term}
}
}
&
\makecell[l]{
\\
{\scriptsize anti-unification:}
\\
{\scriptsize $x,y$ vars}
\\
\scalebox{.9}{
\input{figure/example_int/r_term}
}
\\
{\scriptsize $r \in CPG_E(s,t)$}
\\
{\scriptsize with witnesses}
\\

\input{figure/example_int/witnesses}
}
&
\makecell{
\input{figure/example_int/k_diag}

\\
{\scriptsize we have $k = r \sigma_r \mappingNotation_k$ with:}
\\
\scalebox{.9}{
$
\begin{array}{l}
\sigma_r = 
\left\{\!
\begin{array}{l}
x \mapsto {\scriptstyle\textsf{seq}(\sigma_s(x),\sigma_t(x))},\\
y \mapsto {\scriptstyle\textsf{seq}(\sigma_s(y),\sigma_t(y))}
\end{array}
\!\right\}
\\[0.25cm]
\mappingNotation_k = 
\left\{\!
\begin{array}{l}
\mathfrak{\textcolor{red}{a}} \mapsto {\scriptstyle\textsf{vp}(\mappingNotation_i(\mathfrak{\textcolor{red}{a}}),\mappingNotation_j(\mathfrak{\textcolor{red}{a}}))},\\
\mathfrak{\textcolor{violet}{b}} \mapsto {\scriptstyle\textsf{vp}(\mappingNotation_i(\mathfrak{\textcolor{violet}{b}}),\mappingNotation_j(\mathfrak{\textcolor{violet}{b}}))},\\
\mathfrak{\textcolor{cyan}{c}} \mapsto {\scriptstyle\textsf{vp}(\mappingNotation_i(\mathfrak{\textcolor{cyan}{c}}),\mappingNotation_j(\mathfrak{\textcolor{cyan}{c}}))}
\end{array}
\!\right\}
\end{array}
$
}
}
\\
\hline 
\end{tabular}

%% file: figure/example_int/i_diag.tex
\begin{tabular}{cc}
\makecell[l]{
$i:$
}
&
\raisebox{-.5\height}{
\includegraphics[scale=.45]{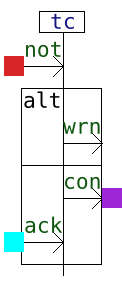}
}
\end{tabular}

%% file: figure/example_int/i_term.tex
\begin{tikzpicture}[every node/.style = {shape=rectangle, align=center}]
\node (o) { \textsf{seq} } [sibling distance=.5cm,level distance=0.5cm]
  child { node (o1) {$\mathfrak{\textcolor{red}{a}}$} }
  child { node (o2) {\textsf{alt}} [sibling distance=1cm,level distance=0.55cm]
    child { node (o21) {$u$} } 
    child { node (o22) {$\textsf{seq}(\mathfrak{\textcolor{violet}{b}},\mathfrak{\textcolor{cyan}{c}})$} }
  }
  ;
  \node[left=-.1cm of o] {$s =$};
\end{tikzpicture}

%% file: figure/example_int/j_diag.tex
\begin{tabular}{cc}
\makecell[l]{
$j:$
}
&
\raisebox{-.5\height}{
\includegraphics[scale=.45]{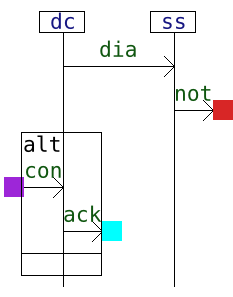}
}
\end{tabular}

%% file: figure/example_int/j_term.tex
\begin{tikzpicture}[every node/.style = {shape=rectangle, align=center}]
\node (o) { \textsf{seq} } [sibling distance=1.5cm,level distance=0.5cm]
  child { node (o1) {\textsf{seq}} [sibling distance=.35cm,level distance=0.55cm]
    child { node (o11) {$v$} }
    child { node (o12) {$\mathfrak{\textcolor{red}{a}}$} }
  }
  child { node (o2) {\textsf{alt}} [sibling distance=1cm,level distance=0.55cm]
    child { node (o21) {$\textsf{seq}(\mathfrak{\textcolor{violet}{b}},\mathfrak{\textcolor{cyan}{c}})$} }
    child { node (o22) {$\varnothing$} }
  }
  ;
  \node[left=-.1cm of o] {$t =$};
\end{tikzpicture}

%% file: figure/example_int/r_term.tex
\begin{tikzpicture}[every node/.style = {shape=rectangle, align=center}]
\node (o) { \textsf{seq} } [sibling distance=1cm,level distance=0.5cm]
  child { node (o1) {\textsf{seq}} [sibling distance=.35cm]
    child { node (o11) {$x$} }
    child { node (o12) {$\mathfrak{\textcolor{red}{a}}$} }
  }
  child { node (o2) {\textsf{alt}} [sibling distance=.5cm]
    child { node (o21) {\textsf{seq}} [sibling distance=.35cm]
      child { node (o211) {$\mathfrak{\textcolor{violet}{b}}$} }
      child { node (o212) {$\mathfrak{\textcolor{cyan}{c}}$} }
    }
    child { node (o22) {$y$} }
  }
  ;
  \node[left=-.1cm of o] {$r =$};
\end{tikzpicture}

%% file: figure/example_int/witnesses.tex
\scalebox{.9}{
$
\begin{array}{lll}
\sigma_s 
&
=
&
\left\{\!
\begin{array}{l}
x \mapsto \varnothing\\
y \mapsto u
\end{array}
\!\right\}
\\[0.25cm]
\sigma_t 
&
=
&
\left\{\!
\begin{array}{l}
x \mapsto v\\
y \mapsto \varnothing
\end{array}
\!\right\}
\end{array}
$
}

%% file: figure/example_int/k_diag.tex
\begin{tabular}{cc}
\makecell[l]{
$k:$
}
&
\raisebox{-.5\height}{
\includegraphics[scale=.45]{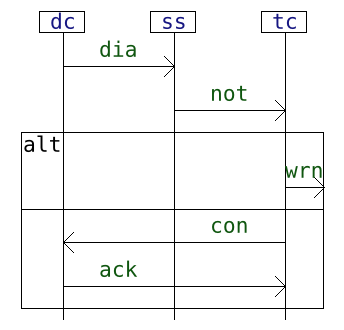}
}
\end{tabular}

%% file: input/2_0_prelim.tex
\section{Preliminaries}
\label{sec:Preliminaries}

\input{input/2_1}

\input{input/2_2}

%% file: input/2_1.tex
\paragraph{Lists.} For any set $X$, we denote by $X^*$ (resp.~$X^n$ for $n \in \mathbb{N}$ and $X^+$) the set of all lists (resp.~lists of length $n$ and non-empty lists) of elements of $X$. $\varepsilon$ denotes the empty list and for any list $\overline{x} \in X^*$ of length $|\overline{x}|$, $\forall i \in [1,|\overline{x}|],~x_i$ denotes the element of $\overline{x}$ at index $i$.
For any $p,q \in X^*$, we say that $p$ is a strict prefix of $q$, denoted by $p \triangleleft q$, iff there exists some $p' \in X^+$ such that $q = p.p'$ where ``$.$'' denotes list concatenation. 
This defines a strict partial order $\triangleleft$ on lists.

\paragraph{Terms.} We adopt notations from \cite{CernaUnital}. Let $\F = (\F_n)_{n \in \mathbb{N}}$ be a family of function symbols indexed by their arity and $\V$ be a countable set of variables. The set of $\F$-terms over $\V$, denoted $\termalg$, is defined inductively by $t ::= x \mid f(t_1,\cdots,t_n)$ where $x \in \V$, $f \in \F_n$ and $t_1,\ldots,t_n \in \termalg^n$.
For any $f \in \F_n$ and any list of $n$ terms $\overline{t} \in \termalg^n$, we may denote the term $f(t_1,\cdots,t_n)$ by $f(\overline{t})$.
Constants are function symbols in $\F_0$. The head symbol of a term is $\head(x)=x$ and $\head(f(\overline{t}))=f$. Given a term $t$, $\var(t)$ is the set of variables occurring in it. If $\var(t) = \emptyset$, $t$ is said to be ground and $\termground$ denotes the set of all ground terms. 
The set of positions of $t$, denoted $\pos(t) \subset \mathbb{N}^*$, is defined by $\pos(x) = \{\varepsilon\}$ and $\pos(f(\overline{t})) = \{\varepsilon\} \cup \bigcup_{i=1}^{|\overline{t}|} \{i.p \mid p \in \pos(t_i)\}$, where the empty list $\varepsilon$ identifies the root position. 
For $p \in \pos(t)$, the subterm at position $p$ is denoted $t|_{p}$.
The \emph{size} of a term $t$ is defined inductively as $\size(x) = 1$, and $\size(f(\overline{t})) = 1+ \sum_{1}^{|\overline{t}|} \size(t_i)$.

\paragraph{Substitutions.}  
A \emph{substitution} $\sigma : \V \to \termalg$ is a mapping 
which differs from the identity only on a finite number of variables. 
The application of $\sigma$ to a term $t$,  written $t\sigma$ using the usual postfix notation, is defined as $x\sigma$ if $t \in \V$, and $f(t_1\sigma, \dots, t_n\sigma)$ if $t = f(t_1, \dots, t_n)$, where $f \in \F_n$. 
The \emph{domain}, \emph{range} and \emph{variable range} of $\sigma$ are resp. defined as $\dom(\sigma) = \{x \mid x\sigma \ne x\}$, $\ran(\sigma) = \{x\sigma \mid x \in \dom(\sigma)\}$, and $\vran(\sigma) = \bigcup_{x \in \dom(\sigma)} \var(x\sigma)$. A substitution can be written as $\{x_1 \mapsto t_1, \dots, x_n \mapsto t_n\}$, listing only the bindings for the variables in its domain.
The identity substitution is denoted as $Id_\V$. 
The \emph{composition} $\sigma\theta$ of substitutions $\sigma$ and $\theta$ is defined by $t\sigma\theta = (t\sigma)\theta$ for any term $t$.

A substitution $\sigma$ is a \emph{renaming} if there exists a substitution $\sigma^{-1}$ such that $\sigma\sigma^{-1} = Id_{\V}$. 

Two terms $t$ and $t'$ are \emph{renaming-equivalent}
if there exists a renaming $\sigma$ with $t\sigma = t'$.

%% file: input/2_2.tex
\paragraph{Equational theory.}
Terms of an algebra $\termalg$ can be subject to properties expressed using equations.
An equation is a couple $(t,t')$ of terms in $\termalg^2$, written as $t \approx t'$. Given a set of equations $E$, the equational theory $=_E$ is the least congruence on $\termalg$ that contains $E$ and is closed under substitutions.

\paragraph{$E$-generalization.}
For two terms $r,t \in \termalg$, we say that $r$ is \emph{less specific} than $t$ modulo a set of equations $E$, written $r \preceq_E t$, if there exists a substitution $\sigma$ such that $r\sigma =_E t$. 
We also say that $r$ is $E$-more general than $t$ as it is possible to transform $r$, via the application of a substitution, into a term that is $E$-equivalent to the more specific $t$.
The relation $\preceq_E$ is a \emph{preorder}, i.e. is reflexive and transitive.
We denote by $\prec_E$ and $\simeq_E$ the corresponding strict and equivalence relations, i.e. $r \prec_E t$ iff $(r \preceq_E t)\wedge(t \not\preceq_E r)$ and $r \simeq_E t$ iff $(r \preceq_E t)\wedge(t \preceq_E r)$. 

A term $r$ is a \emph{$E$-generalization} of two terms $s$ and $t$ if there exist substitutions $\sigma_s$ and $\sigma_t$ such that $r\sigma_s =_E s$ and $r\sigma_t =_E t$, i.e., $r \preceq_E s$ and $r \preceq_E t$. 
The substitutions $\sigma_s$ and $\sigma_t$ are the \emph{witnessing substitutions}.
Intuitively, $r$ captures a common structure of $s$ and $t$ modulo $E$, abstracting some differences by variables. 
For instance, given $E=\{f(x,f(y,z)) \approx f(f(x,y),z)\}$ an $E$-generalization of $s = f(f(a,u),u))$ and $t = f(a, f(v,v))$, is $r = f(a, f(x,x))$ with $\sigma_s = \{x \mapsto u\}$ and $\sigma_t = \{x \mapsto v\}$.
The term $r' = f(a, f(x,y))$ is also an $E$-generalization of $s$ and $t$ but $r'$ is strictly more general than $r$, i.e. $r' \prec_E r$, since $r' \{y \mapsto x\} = r$ and $r\npreceq_E r'$.

A \emph{least general $E$-generalization} ($E$-lgg) $r$ is an $E$-generalization that is \emph{maximal} with respect to~$\preceq_E$, i.e., there exists no generalization $r'$ such that $r \prec_E r'$. 
Dually, a \emph{most general $E$-generalization} ($E$-mgg) reduces $s$ and $t$ to a single variable $x$ (up to renaming), i.e., $r = x$. It does not capture any structure and is mainly of algorithmic interest. 

When $E=\emptyset$ we will remove the $E$- prefix and use the wording lgg, mgg etc., as well as the notations $\preceq$, $\prec$, $\simeq$.
The existence and uniqueness of lggs up to variable renaming has been established, for example, in 
\cite{Unification_revisited}.

%% file: input/3_0_scp_genrlz.tex
\section{Constant-preserving generalization}
\label{sect:constant-preserving-generalization}

\input{input/3_1}

\input{input/3_2}

%% file: input/3_1.tex
Let us introduce a set $\setOfAllSpecialConstants \subseteq \F_0$ of \emph{special constants}. In the following, we generically denote by $\mathfrak{a}, \mathfrak{b}, \mathfrak{c}, \dots$ elements of $\setOfAllSpecialConstants$, and by $u, v, w, \dots$ the \emph{ordinary constants} in $\F_0 \setminus \setOfAllSpecialConstants$. 
For $t \in \termalg$, let $SC(t)$ denote the set of special constants occurring in $t$, and extend this to sets by $SC(\{t_1, \dots, t_n\}) = \bigcup_{i=1}^{n} SC(t_i)$.

For $r,s,t \in \termalg$, $r$ is a \emph{special constant-preserving generalization (\specialConstantPreservingGeneralisation)} of $s$ and $t$ if there exist witnessing substitutions $\sigma_s$ and $\sigma_t$ such that $r$ is a generalization of $s$ and $t$ via $\sigma_s$ and $\sigma_t$ that do not introduce special constants:

\begin{definition}
\label{def:sc_preserving_generalization}
The set of \emph{special constant-preserving (sc-preserving) generalizations} of two terms $s,t \in \termalg$ is defined as:
\[
CPG(s,t) = 
\left\{
r \in \termalg 
\middle|
\begin{array}{l}
\exists \sigma_s,~( r \sigma_s = s) \wedge (SC(\ran(\sigma_s)) = \emptyset)\\
\exists \sigma_t,~( r \sigma_t  = t) \wedge (SC(\ran(\sigma_t)) = \emptyset)\\
\end{array}
\right\}
\]
\end{definition}

Using the above notation, we observe that, since \textbf{(1)} $\sigma_s$ and $\sigma_t$ do not associate  terms with special constants to variables, and that \textbf{(2)} the equalities $r\sigma_s = s$ and $r\sigma_t = t$ hold, this implies that all special constants appearing in $s$ or in $t$ are preserved in $r$:

\begin{lemma}
\label{lemma:non-existence-cg}
For any two terms $s,t \in \termalg$ we have:
\[
CPG(s,t) \neq \emptyset 
~\Rightarrow~ 
\left(
( SC(s) = SC(t) )
~\wedge~
( \forall r \in CPG(s,t),~ SC(r) = SC(s) )
\right)
\]
\end{lemma}

\begin{proof}
Assume $SC(s) \ne SC(t)$. Without loss of generality, let $\mathfrak{a} \in SC(s) \setminus SC(t)$. 
Suppose, for contradiction, that $r \in CPG(s,t)$ witnessing substitutions $\sigma_s$ and $\sigma_t$. Since $r\sigma_t = t$ and $\mathfrak{a} \notin SC(t)$, $\mathfrak{a}$ cannot occur in $r$. But $r\sigma_s = s$ requires $\mathfrak{a} \in SC(\ran(\sigma_s))$, contradicting the sc-preserving condition.\\
Let $r$ in $CPG(s,t)$ with witnessing substitutions $\sigma_s$ and $\sigma_t$. As $s = r\sigma_s$, we have $SC(s) = SC(r\sigma_s) = SC(r) \cup SC(\ran(\sigma_s)) = SC(r)$.
\qed
\end{proof}

In Def.\ref{def:sc-preserving-lgg}, we introduce sc-preserving lggs and mggs (\specialConstantPreservingLgg\, and resp.~\specialConstantPreservingMgg).

\begin{definition}
\label{def:sc-preserving-lgg}
For any $s, t \in \termalg$, we have that $r \in CPG(s,t)$ is an sc-preserving mgg (resp.~lgg) of $s$ and $t$ if $\not\exists~r' \in CPG(s,t)$ s.t.~$r' \prec r$ (resp.~$r \prec r'$) i.e., $r$ is minimal (resp.~maximal) in $CPG(s,t)$ w.r.t.~$\preceq$.
\end{definition}

It transpires that if an \specialConstantPreservingLgg\, exists, then it is an lgg:

\begin{property}\label{prop:sc-preserving-lgg-lgg}
A sc-preserving lgg of $s$ and $t$ in $\termalg$ is an lgg of $s$ and $t$.
\end{property}

\begin{proof}
Let $r$ be an \specialConstantPreservingLgg\, of $s$ and $t$ with witnessing substitutions $\sigma_s$ and $\sigma_t$. Suppose, for contradiction, that a generalization $r'$ of $s$ and $t$ with witnessing substitutions $\sigma_s',\sigma_t'$ exists such that $r \prec r'$, i.e., $r\sigma = r'$ for some substitution $\sigma$. Then
$r\sigma\sigma_s' = r\sigma_s = s$. 
Restricting domains of $\sigma$ 
and $\sigma_s$ 
to $\var(r)$, 
we deduce $\sigma\sigma_s' = \sigma_s$. 
Hence, $SC(\ran(\sigma_s')) \subseteq SC(\ran(\sigma_s))$. 
Since $r\in CPG(s,t)$, $SC(\ran(\sigma_s)) = \emptyset$ and we get $SC(\ran(\sigma_s')) = \emptyset$. A similar reasoning holds by replacing the index $s$ by $t$.
Thus, $r'$ is also sc-preserving, contradicting the maximality of $r$ in $CPG(s,t)$. 
\qed 
\end{proof}

\begin{figure}[t]
    \centering
 \resizebox{.8\textwidth}{!}{%
\input{figure/generalisers_lattice}
 }   
    \caption{Example illustrating generalizations and sc-preserving generalizations}
    \label{fig:generalisers_example_lattice}
\end{figure}
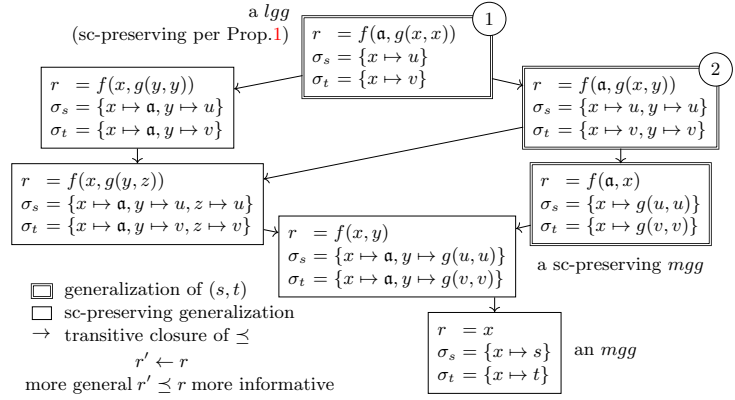

Consider the following example: $s = f(\mathfrak{a}, g(u, u))$ and $t = f(\mathfrak{a}, g(v, v))$. On Fig.\ref{fig:generalisers_example_lattice}, we provide several generalizations of $(s,t)$ and how they are related via the $\preceq$ relation. 
Here, $r = f(\mathfrak{a}, g(x,x))$ (\textcircled{1} on Fig.\ref{fig:generalisers_example_lattice}) is an lgg i.e., it minimally abstracts the differences between $s$ and $t$ and it is sc-preserving i.e., the special constant $\mathfrak{a}$ remains.
As for $r' = f(\mathfrak{a}, g(x,y))$ (\textcircled{2} on Fig.\ref{fig:generalisers_example_lattice}), it is also sc-preserving, but it is not least generalizing. Indeed, $r' \{y \mapsto x\} = r$.

Now consider a different example with $s = f(\mathfrak{a}, g(\mathfrak{b},u))$ and $t = f(\mathfrak{a}, g(v,\mathfrak{b}))$.
Here, any generalization would require a variable to map to a term with special constants: for instance, the candidate $r = f(\mathfrak{a}, g(x, y))$ would require $x\sigma_s = \mathfrak{b}$ and $y\sigma_t = \mathfrak{b}$, thus violating the sc-preservation.

\begin{lemma}
\label{lem:unicity_sc_preserving_lgg}
If a sc-preserving lgg exists, it is unique up to renaming equivalence.
\end{lemma}

\begin{proof}
Let $r$ be a \specialConstantPreservingLgg\, of $s$ and $t$. By Prop.\ref{prop:sc-preserving-lgg-lgg}, $r$ is an lgg, which is unique up to renaming by~\cite{Unification_revisited}.
\qed
\end{proof}

%% file: figure/generalisers_lattice.tex
\begin{tikzpicture}[->,xscale=1.1,yscale=.9]

\node[draw] (mgg) at (0,0) {
$
\begin{array}{lll}
r&=&x\\
\sigma_s&=&\{x \mapsto s\}\\
\sigma_t&=&\{x \mapsto t\}\\
\end{array}
$
};

\node[draw, above left=.25cm and -1.5cm of mgg] (g1) {
$
\begin{array}{lll}
r&=&f(x,y)\\
\sigma_s&=&\{x \mapsto \mathfrak{a}, y \mapsto g(u,u)\}\\
\sigma_t&=&\{x \mapsto \mathfrak{a}, y \mapsto g(v,v)\}\\
\end{array}
$
};

\node[draw, above left=-.5cm and .25cm of g1] (g2) {
$
\begin{array}{lll}
r&=&f(x,g(y,z))\\
\sigma_s&=&\{x \mapsto \mathfrak{a}, y \mapsto u, z \mapsto u\}\\
\sigma_t&=&\{x \mapsto \mathfrak{a}, y \mapsto v, z \mapsto v\}\\
\end{array}
$
};

\node[draw, above=.25cm of g2] (g3) {
$
\begin{array}{lll}
r&=&f(x,g(y,y))\\
\sigma_s&=&\{x \mapsto \mathfrak{a}, y \mapsto u\}\\
\sigma_t&=&\{x \mapsto \mathfrak{a}, y \mapsto v\}\\
\end{array}
$
};

\node[draw,double, above right=-.5cm and .25cm of g1] (scpg2) {
$
\begin{array}{lll}
r&=&f(\mathfrak{a},x)\\
\sigma_s&=&\{x \mapsto g(u,u)\}\\
\sigma_t&=&\{x \mapsto g(v,v)\}\\
\end{array}
$
};

\node[draw,double, above=.25cm of scpg2] (scpg3) {
$
\begin{array}{lll}
r&=&f(\mathfrak{a},g(x,y))\\
\sigma_s&=&\{x \mapsto u,y \mapsto u\}\\
\sigma_t&=&\{x \mapsto v,y \mapsto v\}\\
\end{array}
$
};

\node[draw,double,above=2cm of g1] (lgg) {
$
\begin{array}{lll}
r&=&f(\mathfrak{a},g(x,x))~~~~\\
\sigma_s&=&\{x \mapsto u\}\\
\sigma_t&=&\{x \mapsto v\}\\
\end{array}
$
};

\node[draw,circle,fill=white, above right=-.25cm and -.25cm of lgg] {1};
\node[draw,circle,fill=white, above right=-.25cm and -.25cm of scpg3] {2};

\draw  (lgg) edge (scpg3);
\draw  (lgg) edge (g3);
\draw  (scpg3) edge (scpg2);
\draw  (scpg2) edge (g1);
\draw  (g3) edge (g2);
\draw  (g2) edge (g1);
\draw[->]  (mgg.north |- g1.south) -- (mgg.north);
\draw  (scpg3) edge (g2);

\node[below right=.5cm and -3.5cm of g2, align=left] (leg) {
generalization of $(s,t)$\\
sc-preserving generalization\\
transitive closure of $\preceq$
};
\node[draw,double, above left=-.4cm and .1cm of leg] {$~$};
\node[draw, above left=-.775cm and .1cm of leg] {$~$};
\node[above left=-1.2cm and 0cm of leg] {$\rightarrow$};

\node[below=-.1cm of leg] {
$\begin{array}{lcr}
\!\!\!   &r' \leftarrow r& \!\!\!  
\\
\text{more general }\!\!\!  &r' \preceq r & \!\!\!  \text{ more  informative}
\end{array}
$
};

\node[above left=-.6cm and .1cm of lgg,align=right] {a $lgg$\\{\footnotesize(sc-preserving per Prop.\ref{prop:sc-preserving-lgg-lgg})}};
\node[right=.1cm of mgg] {an $mgg$};
\node[below=.1cm of scpg2] {a sc-preserving $mgg$};

\end{tikzpicture}

%% file: input/3_2.tex
Conflict positions \cite{ALPUENTE2009,ALPUENTE2014,ALPUENTE2022}
are common positions in $s$ and $t$ at which a difference occurs. If there are no such positions then $s=t$.

We distinguish in Def.\ref{def:conflict-pos} conflict positions in which the diverging subterms do or do not contain special constants.

\begin{definition}
\label{def:conflict-pos}
Let $s,t \in \termground$. A position $p \in \pos(s) \cap \pos(t)$ is a \emph{conflict position} ($cpos$) iff $\head(s|_p) \neq \head(t|_p)$ and $\forall~ q \triangleleft p,~ \head(s|_q) = \head(t|_q)$.
Moreover, it is a \emph{failure $cpos$} ($fcpos$) iff $SC(s|_p) \neq \emptyset$ or $SC(t|_p) \neq \emptyset$ and a \emph{solvable $cpos$} ($scpos$) otherwise.

\end{definition}

We denote by $cpos(s,t)$ (resp.~$scpos(s,t)$, $fcpos(s,t)$) the set of $cpos$ (resp.~$scpos$, $fcpos$) of any two terms $s$ and $t$.
The absence of failure conflict positions guarantees the existence of scpg, and conversely: 

\begin{lemma}
\label{lem:failure_conflict_pos_charac}
Given $s$ and $t$ in $\termground$, 
$
(CPG(s,t)=\emptyset)
~\Leftrightarrow~
(fcpos(s,t) \neq \emptyset)
$
\end{lemma}

\begin{proof}
By induction on $N = \size(s)+\size(t)$. If $N = 2$, then $s$ and $t$ are both constant symbols.  
    \begin{itemize}
        \item If $s \ne t$, then the only generalization of $s$ and $t$ is a variable $x$ (up to renaming). Then:
        \begin{itemize}
            \item If $CPG(s,t) = \emptyset$, then $x$ is not sc-preserving. This implies that $SC(s) \ne \emptyset$ or $SC(t) \neq \emptyset$. Hence $\varepsilon$ is a fcpos of $s$ and $t$ so $fcpos(s,t) \neq \emptyset$.
            \item If $fcpos(s,t) \neq \emptyset$ then let's consider one such $fcpos$. As $s$ and $t$ are constants, it can only be $\varepsilon$. We have that $SC(s|_{\varepsilon}) \neq \emptyset$ or $SC(t|_{\varepsilon}) \neq \emptyset$. We deduce that $x$ is not a sc-preserving generalization. Hence $CPG(s,t) = \emptyset$.
        \end{itemize}

        \item If $s=t$, there is no fcpos, whether $s$ is a special constant or not.
        Generalizations of $s$ and $t$ are either a variable $x$ or $s$. If $s$ is (resp. is not) a special constant, then $s$ (resp. $x$) is a sc-preserving generalization and $CPG(s,t) \neq \emptyset$. 
    \end{itemize}
    
\noindent If $N>2$, then $s = f(s_1,\dots,s_n)$ with $n \geq 1$ or $t = g(t_1,\dots,t_m)$ with $m \geq 1$. Let us suppose the left case $s = f(s_1,\dots,s_n)$ (the right case being equivalent).
    
    \begin{itemize}
        \item If $head(t) \neq f$, then the only generalization of $s$ and $t$ (up to variable renaming) is a variable $x$. Let us suppose that there exists a fcpos, then it will be $\varepsilon$ because $head(s) = f \neq head(t)$ and $SC(s)$ or $SC(t)$ will be not empty. It means that $x$ is not sc-preserving. Hence $CPG(s,t) = \emptyset$. 

        Conversely, if $CPG(s,t) = \emptyset$, then $x$ is not an sc-preserving generalization. By definition, it means that $SC(s) \ne \emptyset$ or $SC(t) \ne \emptyset$. In addition to the fact that $f\ne head(t)$, we deduce that $\varepsilon$ is a fcpos of $s$ and $t$. 
        
        \item If $head(t) = f$, let us introduce $t=f(t_1,\cdots,t_n)$. Then one of the following two possibilities applies:
    (1) either for all $i$, $CPG(s_i,t_i) \neq \emptyset$.  By induction hypothesis, no subterm pairs $(s_i,t_i)$ has a fcpos, thus there is no fcpos for $s$ and $t$. For $r_i \in CPG(s_i,t_i)$,  $f(r_1,\ldots,r_n)$ is in $CPG(s,t)$ and $CPG(s,t) \neq \emptyset$;  
        or (2)  for some $i$,  $CPG(s_i,t_i) = \emptyset$. By induction, there exists a fcpos $p_i$ for $s_i$ and $t_i$. Then, $i.p_i$ is a fcpos of $s$ and $t$ with either $SC(s_{|i.p_i})$ or $SC(t_{|i.p_i})$ not empty. Let us suppose there exists $r \in CPG(s,t)$. Either $r$ is a variable, and in this case, $r$ is not sc-preserving since $SC(s_{|i.p_i})$ or $SC(t_{|i.p_i})$ is not empty; or $r$ is of the form $f(r_1,\ldots,r_n)$, and in this case $r_i$ would be in $CPG(s_i,t_i)$. Contradiction with the hypothesis $CPG(s_i,t_i) = \emptyset$. And $CPG(s,t) = \emptyset$. 
    \end{itemize}
\qed 
\end{proof}

%% file: input/4_0_algo_sansE.tex
\section{Algorithm for sc-preserving generalizations}
\label{sec:algo_without_eqs}

\input{input/4_1_algo}

\input{input/4_2_example}

\input{input/4_3_termination}

\input{input/4_4_soundness}

\input{input/4_5_completeness}

%% file: input/4_1_algo.tex
Following~\cite{Cerna_Implem_UnitAU,CernaUnital}, generalization problems are represented as \emph{anti-unification triples} (AUTs), written $x : s \triangleq t$ for ground terms $s, t \in \termground$. The variable $x$ denotes an mgg of $s$ and $t$, witnessed by the substitutions $\{x \mapsto s\}$ and $\{x \mapsto t\}$. The \emph{size} of an AUT is $\size(x : s \triangleq t) = \size(s) + \size(t)$. 
Def.\ref{def:auts} introduces notations on sets of AUTs.

\begin{definition}
\label{def:auts}
Let $A$ be a subset of the set $\setOfAllAUTs = \mathcal{V} \times \mathcal{T}(\mathcal{F})^2$  of all AUTs.
\[
\begin{array}{lcl}
\idx(\activeAUTs) = \cup_{(x:s\triangleq t) \in \activeAUTs} \{x\} 
&
~~
&
\textsf{Pairs}(A) = \cup_{(x:s\triangleq t) \in \activeAUTs} \{(s,t)\}

\\
\sleft(\activeAUTs) = \cup_{(x:s\triangleq t) \in \activeAUTs} \{x \mapsto s\}
&
~~
& 
\sright(\activeAUTs) = \cup_{(x:s\triangleq t) \in \activeAUTs} \{x \mapsto t\}
\end{array}
\]
\end{definition}

A \emph{configuration} is a tuple $c = \langle \activeAUTs | S | \theta | x_0 \rangle$, where $\activeAUTs \subset \setOfAllAUTs$ is a set of so-called \emph{active} AUTs, $S \subset \setOfAllAUTs$ the store of so-called \emph{solved} AUTs, $\theta$ the accumulated index-variable mappings, and $x_0$ the root variable. We denote by $\sleft(c)$ and $\sright(c)$ the substitutions $\sleft(A\cup S)$ and $\sright(A\cup S)$ respectively. 

Starting from an initial configuration $\conf_0 = \langle \{ x_0: s_0 \triangleq t_0 \} | \emptyset | Id_{\V} | x_0 \rangle$ an anti-unification algorithm such as the ones from \cite{ALPUENTE2014,CernaUnital} computes a lgg of $s_0$ and $t_0$ via applying successive transformations from $\conf_0$ until reaching a configuration $c = \langle \emptyset | S | \theta | x_0 \rangle$ in which the set of active AUTs is empty. Once this is done, the lgg is obtained by applying the substitution $\theta$ to the root variable $x_0$ i.e.~$x_0\theta$ is a lgg of $s_0$ and $t_0$.
The corresponding witnessing substitutions are $\sigma_L = \sleft(c)$ and $\sigma_R = \sright(c)$, satisfying $x_0\theta\sigma_L = s_0$ and $x_0\theta\sigma_R = t_0$.

\begin{figure}[h]
\begin{mdframed}
\centering
\input{algos/without_eqs_v3}
\end{mdframed}
\caption{sc-preserving generalization rules (no equations)}
\label{fig:algo_sc-gneraliz_noE}
\end{figure}

On Fig.\ref{fig:algo_sc-gneraliz_noE}, we introduce the inference rules of our sc-preserving anti-unification algorithm.
The rule \textsf{Decompose} replaces in $A$ an AUT $x : f(\overline{s}) \triangleq f(\overline{t})$ with a set of smaller AUTs $x_i : s_i \triangleq t_i$, introducing fresh variables $x_i \not\in \idx(A \cup S) \cup \dom(\theta)$, and extending $\theta$ with $\sigma = \{x \mapsto f(\overline{x})\}$. The rule \textsf{Solve} transfers an AUT from $A$ to $S$ when $\head(s)\neq\head(t)$, no duplicate $(y:s\triangleq t)$ exists in $S$, and both terms are free of special constants ($SC(s)=SC(t)=\emptyset$). The rule \textsf{Recover} reuses a previously solved pair $(y:s\triangleq t)\in S$ by extending $\theta$ with $\{x\mapsto y\}$, enforcing variable sharing across identical subproblems.
It should be noted that if we remove the condition $SC(s)=SC(t)=\emptyset$ in the \textsf{Solve} rule on Fig.\ref{fig:algo_sc-gneraliz_noE}, then these rules correspond to the usual anti-unification rules \cite{ALPUENTE2014}.

%% file: algos/without_eqs_v3.tex
\begin{prooftree}
\hypo{ 
\langle \{ x: f(\overline{s}) \triangleq f(\overline{t}) \} \cup \activeAUTs 
  \mid S \mid \theta \mid x_0 \rangle
}
\infer[left label=$\textsf{(Decompose)}$]1[
$
\left\{
\begin{array}{l}
f \in \F_n \\
\overline{s},\overline{t} \in \termalg^n \\
\overline{x}:n\text{ fresh vars}
\end{array}
\right. 
$
]
{ 
\langle \bigcup_{i \in [1,n]} \{ x_i: s_i \triangleq t_i\} \cup \activeAUTs 
  \mid S \mid \theta
  \{ x \mapsto f(\overline{x})\} \mid x_0 \rangle
}
\end{prooftree}

\vspace*{.5cm}

\begin{prooftree}
\hypo{ 
\langle \{ x: s \triangleq t \} \cup \activeAUTs 
  \mid S \mid \theta \mid x_0 \rangle
}
\infer[left label=$\textsf{(Solve)}$]1[
$
\left\{
\begin{array}{l}
\head(s) \neq \head(t)\\
SC(s) = SC(t) = \emptyset\\
\nexists y ~s.t.~ (y: s \triangleq t) \in S
\end{array}
\right. 
$
]
{ 
\langle \activeAUTs 
  \mid \{ x: s \triangleq t \} \cup S 
  \mid \theta \mid x_0 \rangle
}
\end{prooftree}

\vspace*{.5cm}

\begin{prooftree}
\hypo{ 
\langle \{ x: s \triangleq t \} \cup \activeAUTs 
  \mid S \mid \theta \mid x_0 \rangle
}
\infer[left label=$\textsf{(Recover)}$]1[
$
\left\{
\begin{array}{l}
\exists y ~s.t.~ (y: s \triangleq t) \in S
\end{array}
\right. 
$
]
{ 
\langle \activeAUTs 
  \mid S 
  \mid \theta\{x \mapsto y\} \mid x_0 \rangle
}
\end{prooftree}

%% file: input/4_2_example.tex
\subsection{Examples and relation to conflict positions}

\begin{example}
We give below a derivation trace using the rules of Fig.\ref{fig:algo_sc-gneraliz_noE} from the initial configuration $\langle \{ x_0: f(\mathfrak{a},g(u,u)) \triangleq f(\mathfrak{a},g(v,v))\} \mid \emptyset \mid Id_\V \mid x_0\rangle$, by using the first capital letter of the name of the rules to identify them.
$$\resizebox{.8\textwidth}{!}{
\begin{prooftree}
    \hypo{
    \langle 
    \{ x_0: f(\mathfrak{a},g(u,u)) \triangleq f(\mathfrak{a},g(v,v))\} 
    \mid 
    \emptyset 
    \mid 
    Id_\V 
    \mid 
    x_0
    \rangle
    }
    \infer[left label=(\textsf{D})]1{
    \langle
    \{ x_1: \mathfrak{a}\triangleq \mathfrak{a}, x_2: g(u,u) \triangleq g(v,v) \}
    \mid 
    \emptyset 
    \mid 
    \{ x_0 \mapsto f(x_1,x_2)\}
    \mid 
    x_0
    \rangle 
    }
    \infer[left label=(\textsf{D})]1{
    \langle
    \{x_2: g(u,u) \triangleq g(v,v) \}
    \mid 
    \emptyset 
    \mid 
    \{ x_0 \mapsto f(x_1,x_2)\}\{ x_1 \mapsto \mathfrak{a}\} 
    \mid 
    x_0
    \rangle 
    }
    \infer[left label=(\textsf{D})]1{
    \langle
    \{x_3: u \triangleq v, x_4: u \triangleq v \}
    \mid 
    \emptyset 
    \mid 
    \{ x_0 \mapsto f(\mathfrak{a},x_2)\}\{ x_2 \mapsto g(x_3,x_4)\} 
    \mid 
    x_0
    \rangle 
    }
    \infer[left label=(\textsf{S})]1{
    \langle
    \{x_4: u \triangleq v \}
    \mid 
    \{x_3: u \triangleq v \} 
    \mid \{ x_0 \mapsto f(\mathfrak{a},g(x_3,x_4))\} 
    \mid 
    x_0
    \rangle 
    }
    \infer[left label=(\textsf{R})]1{
    \langle 
    \emptyset 
    \mid 
    \{ x_3: u \triangleq v \} 
    \mid 
    \{ x_0 \mapsto f(\mathfrak{a},g(x_3,x_4))\} \{ x_4 \mapsto x_3\} 
    \mid 
    x_0
    \rangle
    }
\end{prooftree}
}
$$

In the last step, we have $x_0\theta = f(\mathfrak{a},g(x_3,x_3))$ with  $\{ x_3 \mapsto u\}$ and $\{ x_3 \mapsto v\}$ as left and right instantiations.
\end{example}

\begin{example}\label{ex:sc-preserving-noE_fail}
Here, applying the rules of Fig.\ref{fig:algo_sc-gneraliz_noE} does not allow us to empty the set of active AUT, as no sc-preserving generalization exists. 
$$
\resizebox{.8\textwidth}{!}{%
    \begin{prooftree}
        \hypo{\langle \{ x_0: f(\mathfrak{a},g(\mathfrak{b},u)) \triangleq f(\mathfrak{a},g(v,\mathfrak{b}))\} \mid \emptyset \mid Id_\V \mid x_0\rangle}
        \infer[left label=(\textsf{D})]1{\langle\{ x_1: \mathfrak{a}\triangleq \mathfrak{a}, x_2: g(\mathfrak{b},u) \triangleq g(v,\mathfrak{b}) \}\mid \emptyset \mid \{ x_0 \mapsto f(x_1,x_2)\} \mid x_0\rangle }
        \infer[left label=(\textsf{D})]1{\langle\{x_2: g(\mathfrak{b},u) \triangleq g(v,\mathfrak{b}) \}\mid \emptyset \mid \{ x_0 \mapsto f(x_1,x_2)\}\{ x_1 \mapsto \mathfrak{a}\} \mid x_0\rangle }
        \infer[left label=(\textsf{D})]1{\langle\{x_3: \mathfrak{b} \triangleq v, x_4: u \triangleq \mathfrak{b} \}\mid \emptyset \mid \{ x_0 \mapsto f(\mathfrak{a},x_2)\}\{ x_2 \mapsto g(x_3,x_4)\} \mid x_0\rangle }
    \end{prooftree}
    }
$$
\end{example}

The examples above illustrate that all introduced AUTs correspond to pairs of terms $s,t$ located at the same position $p$ in both $s_0$ and $t_0$ and which have the same parent symbols at any position $q$ in $s_0$ and $t_0$ with $q \triangleleft p$.

\begin{lemma}
\label{lem:rules_reveal_common_structure}
There is a derivation from $c_0$ to $c=\langle \{x : s \triangleq t\} \cup \activeAUTs | S | \theta | x_0\rangle$ iff $\exists~p \in pos(s_0) \cap pos(t_0)$ s.t.~$s_0|_p = s$, $t_0|_p = t$, and for all $q \triangleleft p$, $\head(s_0|_q) = \head(t_0|_q)$.
\end{lemma}

\begin{sketchproof*}
This trivially holds for $c=c_0$ with the position being the root $\varepsilon$.
For any other AUT, we may simply remark that only \textsf{Decompose} introduces new AUTs, each step descending one level with identical heads.
\qed 
\end{sketchproof*}

On Def.\ref{def:final_config} we characterize final configurations of our algorithm as those which are irreducible by the rules of Fig.\ref{fig:algo_sc-gneraliz_noE}.

\begin{definition}
\label{def:final_config}
Let  $\conf = \langle \activeAUTs | S | \theta | x_0 \rangle$ be a configuration reachable from $c_0$. $c$ is said to be \emph{final} if $c$ is irreducible by the rules of Fig.\ref{fig:algo_sc-gneraliz_noE}.\\  Moreover if $\activeAUTs = \emptyset$, $\conf$ is said to be \emph{solved} final configuration.
\end{definition}

In Lem.\ref{lem:final_configuration_and_conflicts}, we relate final configurations that are reachable from $c_0$ to the failure conflict positions of the initial terms $s_0$ and $t_0$.

\begin{lemma}
\label{lem:final_configuration_and_conflicts}
For any final configuration $\conf = \langle \activeAUTs | S | \theta | x_0 \rangle$ we have:
\[
\textsf{Pairs}(\activeAUTs) \subseteq \bigcup_{p \in fcpos(s_0,t_0)} \{(s_0|_p,t_0|_p)\}
\]
\end{lemma}

\begin{sketchproof*}
Let $\conf = \langle \activeAUTs | S | \theta | x_0 \rangle$ be a final configuration.
For $x:s \triangleq t$ in $\activeAUTs$, $\head(s) \neq \head(t)$ (otherwise \textsf{Decompose} could be applied) and $SC(s) \neq \emptyset$ or $SC(t) \neq \emptyset$ (otherwise \textsf{Solve} or \textsf{Recover} could be applied). This combined with Lem.\ref{lem:rules_reveal_common_structure} implies that $\exists p \in fcpos(s_0,t_0)$ s.t.~$s = s_0|_p$ and $t = t_0|_p$. 
\qed 
\end{sketchproof*}

%% file: input/4_3_termination.tex
\subsection{Properties}

Termination of our anti-unification algorithm can be directly proved via reasoning on a decreasing multiset measure.

\begin{theorem}
[Termination]
\label{thm:syntactic-termination}
Derivations from $\conf_0$ applying the rules of Fig.\ref{fig:algo_sc-gneraliz_noE} are finite.
\end{theorem}

\begin{proof}
Let us first define, for a configuration $\conf = \langle \activeAUTs | S | \theta | x_0 \rangle$, the multiset~\footnote{ Multisets are denoted $\{\!\!\{\,\cdots\,\}\!\!\}$ when explicitly listing elements by repeating them as many times as their multiplicity.} 
\[
\mu(\conf) = \{\!\!\{ \size(x:s\triangleq t) | x:s\triangleq t \in \activeAUTs \}\!\!\}
\]
which cumulates the sizes of all AUTs in the active set $\activeAUTs$. Let $\conf=\langle \activeAUTs | S | \theta | x_0 \rangle$ and $\conf'=\langle \activeAUTs' | S' | \theta' | x_0 \rangle$ be configurations where $\conf'$ is obtained from $\conf$ by a single application of a rule of Fig.~\ref{fig:algo_sc-gneraliz_noE}. 
We show that\footnote{For two multisets $M$ and $N$, $M >_{mul} N$ iff 
$N$ can be obtained from $M$ by iteratively replacing one element of $M$ by any number of smaller elements.} $\mu(\conf) >_{mul} \mu(\conf')$ by case analysis. If \textsf{Solve} or \textsf{Recover} are applied, some $x:s\triangleq t$  is removed from $A$. As $\activeAUTs' \subsetneq \activeAUTs$, we directly obtain:
    $
        \mu(\conf') = \mu(\conf) \setminus \leftmultiset \size(x:s\triangleq t)\rightmultiset
    $. Therefore, $\mu(\conf) >_{mul} \mu(\conf')$. If \textsf{Decompose} is applied, given $f \in \mathcal{F}$ of arity $n \geq 0$, some $x: f(\overline{s})\triangleq f(\overline{t})\in \activeAUTs$ is replaced by $\bigcup_{i \in [1,n]} \{x_i: s_i\triangleq t_i\}$.
    This gives rise to the following equality:
    $$
\mu(\conf') 
= 
\left(
\mu(\conf) \setminus \leftmultiset \size(x: f(\overline{s})\triangleq f(\overline{t}))\rightmultiset
\right)
~\cup~
\bigcup_{i \in [1..n]}
\leftmultiset \size(x_i:s_i\triangleq t_i) \rightmultiset.
    $$

We get $\mu(\conf) >_{mul} \mu(\conf')$ as for each $i$, 

$\size(x_i:s_i\triangleq t_i) < \size(x: f(\overline{s})\triangleq f(\overline{t}))$.

\noindent As $>_{mul}$ is a well-founded order, any derivation under rules of Fig.~\ref{fig:algo_sc-gneraliz_noE} is finite.

\qed 
\end{proof}

%% file: input/4_4_soundness.tex
If we note $\langle A_0|S_0|\theta_0|x_0\rangle$ an initial configuration $c_0 = \langle x_0:s_0\triangleq t_0|\emptyset|Id_\V|x_0\rangle$ with ground terms $s_0$ and $t_0$, we have $\dom(\sleft(c_0)) = \dom(\sright(c_0)) = \idx(A_0 \cup S_0) =
\var(x_0\theta_0)$, all equal to $\{x_0\}$. Moreover, $\vran(\theta)=\var(x_0\theta) \setminus \{x_0\}$. These properties are invariant on all the reachable configurations:

\begin{lemma}
\label{lem:invariant_in_configurations}
If $c = \langle A|S | \theta | x_0 \rangle$ derives from $c_0$ via the rules of Fig.\ref{fig:algo_sc-gneraliz_noE} then:
\begin{itemize}
    \item $\dom(\sleft(c)) = \dom(\sright(c)) = \idx(A \cup S) = \var(x_0\theta)$;
    \item $\vran(\theta)=\var(x_0\theta) \setminus \{x_0\}$;
    \item $x_0\theta$ is a generalization of $s_0$ and $t_0$ with respective witnessing substitutions 
$\sleft(\conf)$ and $\sright(\conf)$.
\end{itemize}
\end{lemma}

\begin{proof}
Let us reason by induction on the derivations.

$\bullet$ \textbf{Base case:} 
The initial configuration $\conf_0$ is such that $A_0 = \{x_0:s_0 \triangleq t_0\}$, $S_0 = \emptyset$ and $\theta_0 = Id_V$ and thus $\sleft(\conf_0) = \{x_0 \mapsto s_0\}$ and $\sright(\conf_0) = \{x_0 \mapsto t_0\}$.
Therefore:
\begin{itemize}
    \item $\dom(\sleft(\conf_0)) = \dom(\sright(\conf_0)) = \{x_0\}$ and $\idx(A_0 \cup S_0) = \idx(A_0) = \{x_0\}$ and $\var(x_0\theta_0) = \var(x_0) = \{x_0\}$
    \item $\vran(\theta_0)=\var(x_0\theta_0) \setminus \{x_0\}$
    \item $x_0\theta_0=x_0$ is a generalization of $s_0$ and $t_0$ (it is the mgg) with witnessing substitutions $\sleft(\conf_0)$ and $\sright(\conf_0)$
\end{itemize}

$\bullet$ \textbf{Induction:} 
Given a configuration $\conf$, we denote $\sigmal = \sleft(\conf)$ and $\sigmar = \sright(\conf)$. 
Let us suppose the the property holds for a certain configuration $\conf$ reachable from $\conf_0$. Let then suppose that $\conf'$ is reached from $\conf$ by applying a single step.
This step must correspond to the application of either:
\begin{itemize}

    \item \textsf{Solve}: there exist an AUT $a$ s.t.~$A' = A \setminus \{ a \}$ and $S' = S \cup \{a\}$ and we have $\theta'=\theta$, $\sigmal' = \sigmal$ and $\sigmar' = \sigmar$ therefore the three items hold because $A' \cup S' = A \cup S$ and the substitutions are unchanged.
    
    \item \textsf{Recover}: there are two AUTs $x:s\triangleq t$ and $y:s\triangleq t$ s.t.~$a' \in S$ and $A' = A \setminus \{x:s\triangleq t\}$, $\theta' = \theta\{x \mapsto y\}$ (with $x \not\in \dom(\theta)$ and $x \not\in \vran(\theta)$) and all others unchanged.
    Then:
    \begin{itemize}
        \item $\dom(\sigmal') = \dom(\sigmal)\setminus \{ x\}$ and $\dom(\sigmar') = \dom(\sigmar)\setminus \{ x\}$
        \item $\idx(\activeAUTs' \cup S') = (\idx(\activeAUTs \cup S)\setminus \{ x\}$ and $\var(x_0\theta') = \var(x_0\theta)\setminus \{ x\}$
    \end{itemize}
    Thus, by the induction hypothesis, the first item holds.
    As we also have $\vran(\theta') = \vran(\theta)\setminus \{ x\}$, the second item also holds by induction. Then:
    \begin{itemize}
        \item $x\{x \mapsto y\}\sigmal' = y \sigmal' = s = y \sigmal = x\sigmal$
        \item $\forall z \in \var(x_0\theta) \setminus \{x\}, z\{x \mapsto y\} = z$ and thus $z\{x \mapsto y\}\sigmal' = z \sigmal' = z \sigmal$ since $\sigmal'$ and $\sigmal$ differ only on $x$
    \end{itemize}
    Hence $\{x \mapsto y\}\sigmal' = \sigmal$. The previous two points also hold for $\sigmar$ and $\sigmar'$. Therefore, as $\var(x_0\theta') = \dom(\sigmal') = \dom(\sigmar')$ by the first item we have:
    \begin{itemize}
        \item $x_0\theta'\sigmal' = x_0\theta\{x \mapsto y\}\sigmal' = x_0\theta\sigmal$
        \item $x_0\theta'\sigmar' = x_0\theta\{x \mapsto y\}\sigmar' = x_0\theta\sigmar$
    \end{itemize}
    Thus, by the induction hypothesis, the third item holds

    \item \textsf{Decompose}: an AUT $a = x:f(\overline{s})\triangleq f(\overline{t}) \in A$ is replaced by $n \geq 0$ AUTs $a_i = x_i:s_i\triangleq t_i$ with $n$ fresh variables $\overline{x}$ so that:
    \begin{itemize}
        \item $A' = (A \setminus \{a\}) \cup \{a_1,\cdots,a_n\}$ and $S'=S$
        \item $\theta' = \theta\{x \mapsto f(\overline{x})\}$
        \item $\sigmal' = (\sigmal \setminus\{x\mapsto f(\overline{x})\}) \cup \{x_1\mapsto s_1,\ldots,x_n\mapsto s_n\}$
        \item $\sigmar' = (\sigmar \setminus\{x\mapsto f(\overline{x})\}) \cup \{x_1\mapsto t_1,\ldots,x_n\mapsto t_n\}$
    \end{itemize}
    We then have:
    \begin{itemize}
        \item $\idx(\activeAUTs'\cup S') = (\idx(\activeAUTs \cup S)\setminus \{ x\} )\cup \{x_1,\hdots,x_n \}$
        \item $\dom(\sigmal') = (\dom(\sigmal)\setminus \{ x\}) \cup \{x_1,\hdots,x_n \}$
        \item $\dom(\sigmar')) = (\dom(\sigmar)\setminus \{ x\}) \cup \{x_1,\hdots,x_n \}$
        \item $\vran(\theta') = (\vran(\theta)\setminus \{ x\}) \cup \{x_1,\hdots,x_n \}$
        \item $\var(x_0\theta') = (\var(x_0\theta)\setminus \{ x\}) \cup \{x_1,\hdots,x_n \}$
    \end{itemize}
    Thus, by the induction hypothesis, the first and second items hold. Also:
    \begin{itemize}
        \item $x\{x \mapsto f(\overline{x})\}\sigmal' = f(\overline{x})\sigmal' = f(x_1\sigmal',\cdots,x_n\sigmal') = f(\overline{s}) = x\sigmal$
        \item $\forall z \in \var(x_0\theta) \setminus \{x\}, z\{x \mapsto f(\overline{x})\} = z$ and thus $z\{x \mapsto f(\overline{x})\}\sigmal' = z \sigmal' = z \sigmal$ since $\sigmal'$ and $\sigmal$ differ only on $x$
    \end{itemize}
   Hence, $\{x \mapsto f(\overline{x})\}\sigmal' = \sigmal$. The previous two points also hold for $\sigmar$ and $\sigmar'$. Therefore, as $\var(x_0\theta') = \dom(\sigmal') = \dom(\sigmar')$ by the first item we have:
    \begin{itemize}
        \item $x_0\theta'\sigmal' = x_0\theta\{x \mapsto f(\overline{x})\}\sigmal' = x_0\theta\sigmal$
        \item $x_0\theta'\sigmar' = x_0\theta\{x \mapsto f(\overline{x})\}\sigmar' = x_0\theta\sigmar$
    \end{itemize}
    Thus, by the induction hypothesis the third item holds.
\end{itemize}
\end{proof}

\begin{lemma}[SC-preservation]\label{lem:noeq_sc_preservation}
Let $\conf$ be a solved configuration derived from $\conf_0$ by successive applications of the rules of Fig.~\ref{fig:algo_sc-gneraliz_noE}, then $x_0\theta \in CPG(s_0,t_0)$.
\end{lemma}

\begin{proof}
\noindent 
By Lem.\ref{lem:invariant_in_configurations}, $x_0\theta$ is a generalization of $s_0$ and $t_0$. Let us prove that it is sc-preserving.
By Lem.\ref{lem:invariant_in_configurations},
$\dom(\sleft(\conf))=\dom(\sright(\conf))=\idx(\activeAUTs\cup S)=\idx(S)$ since $\activeAUTs$ is empty for a solved configuration. 
Each binding in $\sleft(\conf)$ (resp.\ $\sright(\conf)$) is of the form $x\mapsto s$ (resp.\ $x\mapsto t$) with $x:s\triangleq t\in S$.
As by construction, the set of solved AUTs only contains special-constant–free AUTs, we have
$SC(\ran(\sleft(\conf)))=SC(\ran(\sright(\conf)))=\emptyset$ so that $x_0\theta\in CPG(s_0,t_0)$.
\qed
\end{proof}

\begin{theorem}
[Soundness]
\label{th:noeq_soundness}
Let $\conf$ be a solved configuration derived from $\conf_0$ by successive applications of the rules of Fig.~\ref{fig:algo_sc-gneraliz_noE}, then $x_0\theta$ is a \specialConstantPreservingLgg\, of $s_0$ and $t_0$.
\end{theorem}

\begin{proof}
Let $\conf = \langle \emptyset | S | \theta |  x_0\rangle$ reached from $\conf_0$. By Lem.\ref{lem:noeq_sc_preservation}, $x_0\theta\in CPG(s_0,t_0)$. Assume that there is an sc-preserving lgg $r$ and a substitution $\rho$ with $x_0\theta\rho=r$ and $\rho$ not a renaming. Take $\dom(\rho)=\var(x_0\theta)$. Since $\rho$ is not a renaming, there exists $y$ in $\dom(\rho)$ that verifies either \textbf{(1)} $\exists y'\in \V\setminus \{y\},  y\rho = y'\rho \in \V $, i.e. $\rho$ merges at least two variables occurring in $x_0\theta$ or \textbf{(2)} $y\rho\notin\V$.

Case \textbf{(1)}. By Lem.\ref{lem:invariant_in_configurations}, the active set being empty, $\idx(S)=\var(x_0\theta)=\dom(\rho)$, so $y$ and $y'$ index solved AUTs in $S$: $y:s\triangleq t\in S$ and $y':s'\triangleq t'\in S$. 
By Lem.\ref{lem:rules_reveal_common_structure} there are solvable positions $p,p'$ with $s_0|_{p}=s,t_0|_{p}=t$ and $s_0|_{p'}=s',t_0|_{p'}=t'$. 
Now $r|_{p}=(x_0\theta\rho)|_{p}=(x_0\theta)|_{p}\rho=y\rho$ and similarly $r|_{p'}=y'\rho$.
Given that we know that $y\rho = y'\rho$, this implies $r|_{p}=r|_{p'}$. 
Since $r\in CPG(s_0,t_0)$ there are witnesses $\sigma_{L},\sigma_{R}$ with $r\sigma_{L}=s_0$ and $r\sigma_{R}=t_0$, hence $s_0|_{p}=r|_{p}\sigma_{L}=r|_{p'}\sigma_{L}=s_0|_{p'}$ and similarly $t_0|_{p}=r|_{p}\sigma_{R}=r|_{'}\sigma_{R}=t_0|_{p'}$. 
We deduce $s=s'$ and $t=t'$. 
Thus $S$ contains two solved AUTs with identical subproblems, which \textsf{Recover} forbids. Contradiction.

Case \textbf{(2)}. Again $\idx(S)=\var(x_0\theta)$ so $y\in\idx(S)$ and $y:s\triangleq t\in S$. By Lem.\ref{lem:rules_reveal_common_structure} there is a solvable position $p$ with $s_0|_p=s,t_0|_p=t$. Now $r|_p=(x_0\theta\rho)|_p=(x_0\theta)|_p\rho=y\rho$, which by hypothesis is a non-variable term. 
But $r\in CPG(s_0,t_0)$, so its witnessing substitutions $\sigma_{L},\sigma_{R}$ satisfy $r\sigma_{L}=s_0$ and $r\sigma_{R}=t_0$. 
Applying $\sigma_{L}$ and $\sigma_{R}$ to $r|_p=y\rho$ yields $s$ and $t$, respectively; hence $\head(s)=\head(t)=\head(y\rho)$. 
Thus the solved AUT $y:s\triangleq t$ has identical root symbols which \textsf{Solve} forbids. Contradiction. 
\qed 
\end{proof}

%% file: input/4_5_completeness.tex
\begin{theorem}[Completeness]
\label{th:algo_noeq_completeness}
For any $r,s_0,t_0 \in \termground$, if $r$ is a sc-preserving lgg of $s_0$ and $t_0$ then there is a derivation from $\conf_0$ to a final configuration $\langle \emptyset | S | \theta | x_0 \rangle$ such that $r \simeq x_0\theta$.
\end{theorem}

\begin{proof}
Let us suppose $r$ is an scplgg of $s_0$ and $t_0$. Then $CPG(s_0,t_0) \neq \emptyset$ so by Lem.\ref{lem:failure_conflict_pos_charac}, we must have $fcpos(s_0,t_0) = \emptyset$.
Then by Lem.\ref{lem:final_configuration_and_conflicts}, for any final configuration $c = \langle \activeAUTs| S|\theta|x_0\rangle$ we must have $\activeAUTs = \emptyset$. 
Because the algorithm is terminating (Th.\ref{thm:syntactic-termination}) one such configuration exists.
Then by Th.\ref{th:noeq_soundness}, we must have that $x\theta$ is a sclgg of $s_0$ and $t_0$.
Finally, by Lem.\ref{lem:unicity_sc_preserving_lgg} we have $r \simeq x_0\theta$.
\qed
\end{proof}

Let us remark that, for any $\conf = \langle A | S | \theta | x_0 \rangle$, if $\exists x : s \triangleq t \in A$ s.t.~$SC(s) \neq SC(t)$, then, in any $\conf'$ reachable from $\conf$, $A'$ cannot be empty.
Indeed, if $\head(s) \neq \head(t)$, only $\textsf{Solve}$ may remove $x : s \triangleq t$ from $A$ and it cannot be applied.
If $\head(s) = \head(t)$ with $s = f(\overline{s})$ and $t = f(\overline{t})$ then $\exists i \in [1,|\overline{s}|]$ s.t.~$SC(s_i) \neq SC(t_i)$. Then, applying $\textsf{Decompose}$ results in adding a $y : s_i \triangleq t_i$ with $SC(s_i) \neq SC(t_i)$ to $A'$. By induction, from $\conf'$, no successor $\conf''$ can have $A'' = \emptyset$.
As a result, one cannot find a scpg by exploring successors of $\conf$.
Thus, we can pre-emptively prune the exploration of configurations by applying the higher priority $\textsf{Fail}$ rule from Fig.\ref{fig:fail}. Its use will speed up the detection of failure cases (see Sec.\ref{sec:application_to_int_compo}).

\begin{figure}[t]
\begin{mdframed}
\centering
\input{algos/fail_rule}
\end{mdframed}
\caption{ $\textsf{(Fail)}$ a derived rule to complete rules of Fig.\ref{fig:algo_sc-gneraliz_noE}}
\label{fig:fail}
\end{figure}

%% file: algos/fail_rule.tex
\begin{prooftree}
\hypo{ 
\langle \{ x: s \triangleq t \} \cup \activeAUTs 
  \mid S \mid \theta \mid x_0 \rangle
}
\infer[left label=$\textsf{(Fail)}$]1[
$
\left\{
\begin{array}{l}
SC(s) \neq SC(t)
\end{array}
\right. 
$
]
{ 
\bot
}
\end{prooftree}

%% file: input/5_0_equational_scp_genrlz.tex
\section{Equational constant-preserving generalization}
\label{sec:extension_modulo_E}

\input{input/5_1}

\input{input/5_2}

%% file: input/5_1.tex
Sc-preserving anti-unification modulo an equational theory $E$ requires $E$ to be sc-preserving as per Def.\ref{def:eqn-sc-preserving-generalization}.
Otherwise, applying the equations could erase the special-constants (e.g., $f(a_f,x) \approx a_f$ with an absorbing element $a_f$ of $f$).

\begin{definition}
\label{def:eqn-sc-preserving-generalization}
Given an equational theory $E$:
\begin{itemize}
    \item $E$ is \emph{sc-preserving} if $\forall s,t\in\termalg,~ s =_E t \Rightarrow SC(s) = SC(t)$
    \item the \emph{sc-preserving $E$-generalizations} ($E$-\specialConstantPreservingGeneralisation) of terms $s,t \in \termalg$ are:
\[ CPG_E(s,t) = \left\{ r \in \termalg \middle| \begin{array}{l} \exists \sigma_s,~( r \sigma_s =_E s) \wedge (SC(\ran(\sigma_s)) = \emptyset)\\ \exists \sigma_t,~( r \sigma_t =_E t) \wedge (SC(\ran(\sigma_t)) = \emptyset)\\ \end{array} \right\} \]
\end{itemize}
\end{definition}

Similarly as we did for $CPG$ in Lem.\ref{lemma:non-existence-cg}, we characterize the preservation of special constants by $CPG_E$ in Lem.\ref{lem:cpg_e_preserves_sc}.

\begin{lemma}
\label{lem:cpg_e_preserves_sc}
For any sc-preserving theory $E$ and $\forall s,t \in \termalg$ we have:
\[
CPG_E(s,t) \neq \emptyset 
~\Rightarrow~ 
\left(
( SC(s) = SC(t) )
~\wedge~
( \forall r \in CPG_E(s,t),~ SC(r) = SC(s) )
\right)
\]
\end{lemma}

\begin{proof}   
Assume that $r \in CPG_E(s,t)$ exists with witnessing substitutions $\sigma_s$ and $\sigma_t$.
Since $E$ is sc-preserving, we have $SC(r\sigma_s) = SC(s)$ and  $SC(r\sigma_t) = SC(t)$. But $SC(r\sigma_s) = SC(r\sigma_t) = SC(r)$ because $\sigma_s$ and $\sigma_t$ introduce no special constants. Hence, $SC(s) = SC(t)$. The second part is obvious.
\qed
\end{proof}

Unlike the \specialConstantPreservingLgg\, of Def.\ref{def:sc-preserving-lgg}, the $E$-\specialConstantPreservingLgg\, of Def.\ref{def:sc-preserving-E-lgg} are not necessarily unique up to renaming.
Consider for instance
$s = f(\mathfrak{a}, g(u, u))$ and $t = f(g(v, v),\mathfrak{a})$ with $E = \{f(x,y) \approx f(y,x)\}$.
Here both $r_1 = f(\mathfrak{a}, g(x,x))$ and $r_2 = f(g(x,x),\mathfrak{a})$ are $E$-\specialConstantPreservingLgg s of $s$ and $t$.
Yet, we may remark that we have $r_1 =_E r_2$.

\begin{definition}
\label{def:sc-preserving-E-lgg}
For any $s, t \in \termalg$, we have that $r \in CPG_E(s,t)$ is an $E$-\specialConstantPreservingMgg\, (resp.~$E$-\specialConstantPreservingLgg) of $s$ and $t$ if $\not\exists~r' \in CPG_E(s,t)$ s.t.~$r' \prec_E r$ (resp.~$r \prec_E r'$) i.e., $r$ is minimal (resp.~maximal) in $CPG_E(s,t)$ w.r.t.~$\preceq_E$.
\end{definition}

With Prop.\ref{prop:sc-preserving-E-lgg-lgg}, we adapt
Prop.\ref{prop:sc-preserving-lgg-lgg} to the case modulo equations.

\begin{property}
\label{prop:sc-preserving-E-lgg-lgg}
A sc-preserving $E$-lgg of $s$ and $t$ in $\termalg$ is an $E$-lgg.
\end{property}
\begin{proof}
Let $r$ be an sc-preserving $E$-lgg of $s$ and $t$ with witnessing substitutions $\sigma_s$ and $\sigma_t$. Suppose, for contradiction, that a $E$-generalization $r'$ of $s$ and $t$ with witnessing substitutions $\sigma_s',\sigma_t'$ exists such that $r \prec_E r'$, i.e., $r\sigma =_E r'$ for some substitution $\sigma$. Then
$r\sigma\sigma_s' =_E r\sigma_s =_E s$ and $r\sigma\sigma_t' =_E r\sigma_t =_E t$

Restricting domains of $\sigma, \sigma_t$ and $\sigma_s$ (resp. $\sigma_t'$ and $\sigma_s'$) to $\var(r)$ (resp. $\var(r')$), we deduce $\sigma\sigma_s' = \sigma_s$ and $\sigma\sigma_t' = \sigma_t$. Hence, $SC(\ran(\sigma_t')) \subseteq SC(\ran(\sigma_t))$ and $SC(\ran(\sigma'_t))\subseteq SC(\ran(\sigma_t))$. Since $r$ is a sc-preserving $E$-generalization, $SC(\ran(\sigma_s)) = SC(\ran(\sigma_t)) = \emptyset$, hence $SC(\ran(\sigma_s')) = \emptyset$ and $SC(\ran(\sigma_t')) = \emptyset$.
Thus, $r'$ is also sc-preserving, contradicting the maximality of $r$ as an sc-preserving $E$-generalization. Therefore, $r$ is an $E$-lgg of $s$ and $t$. \hfill $\square$
\end{proof}

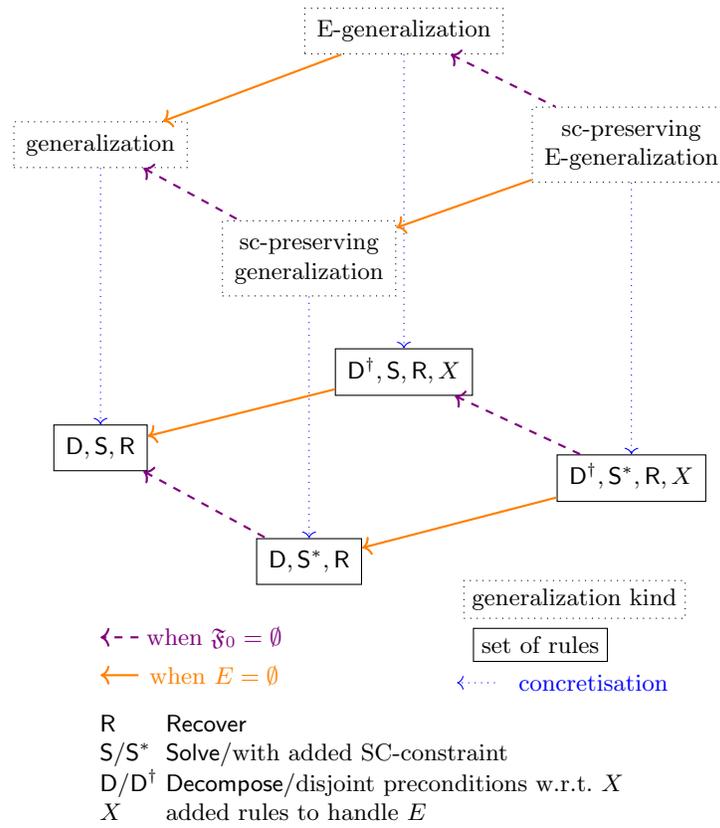
\begin{figure}
    \centering
    \resizebox{0.8\linewidth}{!}{%
        \input{figure/algos_generalizers_relations}
    }   
    \caption{Relating generalizations}
\label{fig:diagram_generalisation_kinds}
\end{figure}

Fig.\ref{fig:diagram_generalisation_kinds} illustrates the relationships between standard generalization~\cite{ALPUENTE2009,ALPUENTE2014,ALPUENTE2022} and our sc-preserving ($E$-)generalization. 

The dashed horizontal arrows indicate that our formulation naturally subsumes standard generalization when the set of special constants is empty. 
Similarly, the plain horizontal arrows show that in the absence of equational axioms, the problem reduces to syntactic generalization.
The top and bottom layers resp.~correspond to the problem definitions and the inference rules used to solve them.

We update the standard rules $\{\textsf{D}, \textsf{S}, \textsf{R}\}$, substituting \textsf{S} with a $\textsf{S}^*$ that enforces constant preservation\footnote{
Our \textsf{Solve} rule from Fig.\ref{fig:algo_sc-gneraliz_noE} is denoted by $\textsf{S}^*$ on Fig.\ref{fig:diagram_generalisation_kinds} as it differs from the standard \textsf{Solve} rule from \cite{ALPUENTE2014,CernaUnital} by the added condition $SC(s)=SC(t)=\emptyset$.

}.
$E$-generalization is implemented in \cite{ALPUENTE2014,CernaUnital} by introducing additional ad-hoc rules to handle specific theories (e.g.~an operator $f$ being associative, commutative, or having a neutral element, etc.) and combinations of theories (e.g.~$f$ being both associative and having a neutral element, etc.). On Fig.\ref{fig:diagram_generalisation_kinds}, we represent this via the set $X$ of additional rules, and the use of $\textsf{D}^\dagger$ instead of $\textsf{D}$ as we do not apply the standard \textsf{Decompose} rule if a rule in $X$ can be applied.
To achieve sc-preserving $E$-generalization, we follow the same principle, the only difference being the use of the special $\textsf{S}^*$ rule.

%% file: figure/algos_generalizers_relations.tex
\begin{tikzpicture}

\node[draw,dotted] (th_general) at (-7,6.5) {
$
\begin{array}{c}
\text{generalization}
\end{array}
$
};

\node[draw,dotted] (th_e_general) at (-3,8) {
$
\begin{array}{c}
\text{E-generalization}
\end{array}
$
};

\node[draw,dotted] (th_scp_e_general) at (0,6.5) {
$
\begin{array}{c}
\text{sc-preserving}
\\
\text{E-generalization}
\end{array}
$
};

\node[draw,dotted] (th_scp_general) at (-4.25,5) {
$
\begin{array}{c}
\text{sc-preserving}
\\
\text{generalization}
\end{array}
$
};

\node[draw] (rules_general) at (-7,2.5) {
$
\begin{array}{c}
\textsf{D}, \textsf{S}, \textsf{R}
\end{array}
$
};

\node[draw]  (rules_scp_general) at (-4.25,1) {
$
\begin{array}{c}
\textsf{D}, \textsf{S}^*, \textsf{R}
\end{array}
$
};

\node[draw]  (rules_e_general) at (-3,3.5) {
$
\begin{array}{c}
\textsf{D}^\dagger, \textsf{S}, \textsf{R}, X
\end{array}
$
};

\node[draw]  (rules_scp_e_general) at (0,2.1) {
$
\begin{array}{c}
\textsf{D}^\dagger, \textsf{S}^*, \textsf{R}, X
\end{array}
$
};

\draw[<-,violet,thick,dashed]  (th_general) edge (th_scp_general);
\draw[<-,violet,thick,dashed]  (rules_general) edge (rules_scp_general);
\draw[<-,violet,thick,dashed]  (th_e_general) edge (th_scp_e_general);
\draw[<-,violet,thick,dashed]  (rules_e_general) edge (rules_scp_e_general);
\draw[->,orange,thick]  (th_e_general) edge (th_general);
\draw[->,orange,thick]  (th_scp_e_general) edge (th_scp_general);
\draw[<-,orange,thick]  (rules_scp_general) edge (rules_scp_e_general);
\draw[<-,orange,thick]  (rules_general) edge (rules_e_general);
\draw[blue,dotted,->]  (th_scp_general) edge (rules_scp_general);
\draw[blue,dotted,->]  (th_general) edge (rules_general);
\draw[blue,dotted,->]  (th_e_general) edge (rules_e_general);
\draw[blue,dotted,->]  (th_scp_e_general) edge (rules_scp_e_general);

\node[violet] at (-5.5,0) {when $\setOfAllSpecialConstants=\emptyset$};
\draw[dashed,->,thick,violet] (-6.5,0) -- (-7,0);

\node[orange] at (-5.5,-.5) {when $E=\emptyset$};
\draw[orange,->,thick] (-6.5,-.5) -- (-7,-.5);

\node[align=left] at (-3.5,-1.75) {
$
\begin{array}{ll}
\textsf{R} & \textsf{Recover}\\
\textsf{S}/\textsf{S}^* & \textsf{Solve}/\text{with added SC-constraint}\\
\textsf{D}/\textsf{D}^\dagger & \textsf{Decompose}/\text{disjoint preconditions w.r.t.}~X\\
X & \text{added rules to handle }E\\
\end{array}
$
};

\node[draw,dotted] at (-.75,.5) {generalization kind};
\node[draw] at (-1.2,-.1) {set of rules};

\node[blue] at (-.5,-.6) {concretisation};
\draw[blue,->,dotted] (-1.8,-.6) -- (-2.3,-.6);
\end{tikzpicture}

%% file: input/5_2.tex
 \begin{figure}[h]
 \vspace*{.25cm}
\begin{mdframed}
\centering
\input{algos/rule_decompose_C}
\end{mdframed}
\caption{Generalization rule for a commutative, non-associative function symbol~\cite{ALPUENTE2014}}
\label{fig:algo_rule_decompose_C}
 \end{figure}

On Fig.\ref{fig:algo_rule_decompose_C}, we propose one such ad-hoc rule. If we have $X=\{\textsf{Decompose}_C\}$ as the $X$ of Fig.\ref{fig:diagram_generalisation_kinds} then the rule \textsf{Decompose} of Fig.\ref{fig:algo_sc-gneraliz_noE} would not be applied if the shared head symbol $f$ is commutative but $\textsf{Decompose}_C$ would be applied instead.

\begin{example} Reusing Ex.\ref{ex:sc-preserving-noE_fail}, 
we now assume that 
$g$ is commutative and non-associative.
By applying the rules of Fig.\ref{fig:algo_sc-gneraliz_noE} together with rule of Fig.\ref{fig:algo_rule_decompose_C}, an sc-preserving $E$-generalization exists: we have $x_0\theta = f(\mathfrak{a},g(\mathfrak{b},x_4))$ with $\{ x_4 \mapsto u\}$ and $\{ x_4 \mapsto v\}$ as left and right instantiations.
$$
\resizebox{.8\textwidth}{!}{%
    \begin{prooftree}
        \hypo{\langle \{ x_0: f(\mathfrak{a},g(\mathfrak{b},u)) \triangleq f(\mathfrak{a},g(v,\mathfrak{b}))\} \mid \emptyset \mid Id_\V \mid x_0\rangle}
        \infer[left label=(\textsf{D})]1{\langle\{ x_1: \mathfrak{a}\triangleq \mathfrak{a}, x_2: g(\mathfrak{b},u) \triangleq g(v,\mathfrak{b}) \}\mid \emptyset \mid \{ x_0 \mapsto f(x_1,x_2)\} \mid x_0\rangle }
        \infer[left label=(\textsf{D})]1{\langle\{x_2: g(\mathfrak{b},u) \triangleq g(v,\mathfrak{b}) \}\mid \emptyset \mid \{ x_0 \mapsto f(x_1,x_2)\}\{ x_1 \mapsto \mathfrak{a}\} \mid x_0\rangle }
        \infer[left label=(\textsf{D}$_{\text{C}}$)]1{\langle\{x_3: \mathfrak{b} \triangleq \mathfrak{b}, x_4: u \triangleq v \}\mid \emptyset \mid \{ x_0 \mapsto f(\mathfrak{a},x_2)\}\{ x_2 \mapsto g(x_3,x_4)\} \mid x_0\rangle }
         \infer[left label=(\textsf{D})]1{\langle\{x_4: u \triangleq v \}\mid \emptyset \mid \{ x_0 \mapsto f(\mathfrak{a},x_2)\}\{ x_2 \mapsto g(x_3,x_4)\}\{ x_3 \mapsto \mathfrak{b}\} \mid x_0\rangle }
          \infer[left label=(\textsf{S})]1{\langle \emptyset \mid  \{x_4: u \triangleq v \} \mid \{ x_0 \mapsto f(\mathfrak{a},x_2)\}\{ x_2 \mapsto g(x_3,x_4)\}\{ x_3 \mapsto \mathfrak{b}\} \mid x_0\rangle }
    \end{prooftree}
    }
$$
\end{example}

%% file: algos/rule_decompose_C.tex
\bigskip
\begin{prooftree}
\hypo{ 
\langle \{ x: f(\overline{s}) \triangleq f(\overline{t}) \} \cup \activeAUTs 
  \mid S \mid \theta \mid x_0 \rangle
}
\infer[left label=$(\textsf{Decompose}_{\text{C}})$]1[
$
\left\{
\begin{array}{l}
f \in \F_2 \\
\overline{s},\overline{t} \in \termalg^2 \\
\overline{x}:2\text{ fresh vars}\\
i\in [1,2] \\
\!\sigma \!=\! \{ x \mapsto\! f(\overline{x}) \}
\end{array}
\right.
$
]
{ 
\langle 
\left\{
\begin{array}{l}
x_1:s_1 \triangleq t_i,\\
x_2:s_2 \triangleq t_{(i \text{ mod } 2)+1}
\end{array}
\right\}
\cup \activeAUTs 
\mid S \mid \theta\sigma \mid x_0 \rangle
}
\end{prooftree}

%% file: input/7_0_composition.tex
\section{Application to the composition of interactions}
\label{sec:application_to_int_compo}

\input{input/7_1_pb}

\input{input/7_2_soundness}

\input{input/7_3_exp}

%% file: input/7_1_pb.tex
Recall the interaction composition problem discussed in Fig.~\ref{fig:example_interaction_composition}.
Let $\mathcal{L}$ and $\M$ be the universes of lifelines and messages.
Atomic actions are  
$\act=\{\,l!m,\; l?m \mid l\in\mathcal{L},\; m\in\M\,\}$,  
where for any $a\in\act$, $\Delta(a)\in\{!,?\}$ and $\messageOfAction(a)\in\M$ denote its type and message.
Value-passings form the set $\vps=\{\vp(l_1,m,l_2)\mid  l_1,l_2\in\mathcal{L}, m\in\M,\; l_1\neq l_2\}$.  
Two actions $u,v\in\act$ are \emph{compatible} iff there exist $l_1,l_2\in\mathcal{L}$ and $m\in\M$ with $l_1\neq l_2$ such that $\{u,v\}=\{l_1!m,l_2?m\}$; for such actions, we may abusively write $\vp(u,v)$ (resp. $\vp(v,u)$) for $\textsf{vp}(l_1,m,l_2)$.
The set of interactions is $\inter=\mathcal{T}(\mathcal{F})$, with signature 
$\F=(\F_n)_{n\in[0,2]}$ defined as:

\[
\F_0 = \act \cup \vps \cup \{\varnothing\}, \quad
\F_1 = \{\loopS\}, \quad
\F_2 = \{\seq, \alt, \para\}.
\]

\noindent For any $i \in \mathbb{I}$, $\lfs(i) \subset \mathcal{L}$ denotes the lifelines occurring in $i$, defined by $\lfs(l!m)=\lfs(l?m)=\{l\}$, $\lfs(\vp(l_1,m,l_2))=\{l_1,l_2\}$, $\lfs(\loopS(i))=\lfs(i)$, and $\lfs(f(i,j))=\lfs(i)\cup\lfs(j)$ for $f\in\{\seq,\alt,\para\}$. Semantically, $\seq$, $\alt$, and $\para$ are associative, $\alt$ and $\para$ commutative, and $\varnothing$ is the unit for $\seq$ and $\para$. Let $E$ denote the set of equations characterizing these properties.

Given a set $\mathfrak{F}_0$ of special constants called gates, a \emph{tagging} $\gamma$ of $i,j \in \mathbb{I}$ is a set such that: \textbf{(1)} each $g \in \gamma$ is a tuple $(p_i,p_j,\mathfrak{a})$ with $\mathfrak{a} \in \mathfrak{F}_0$, $p_i \in pos(i)$, $p_j \in pos(j)$, and $i|_{p_i}$, $j|_{p_j}$ are compatible actions; \textbf{(2)} for $g,g' \in \gamma$ with $g'=(p_i',p_j',\mathfrak{b})$, we have $\mathfrak{a}=\mathfrak{b}$ implies $(i|_{p_i} = i|_{p_i'})\wedge (j|_{p_j} = j|_{p_j'})$. Condition \textbf{(1)} ensures each tag relates compatible actions via a gate, while \textbf{(2)} enforces uniqueness of actions and gates. We denote by $\Gamma(i,j)$ the set of all taggings of $i$ and $j$.

Given interactions $i,j,k \in \mathbb{I}$, $(i,j) \overset{\gamma}{\propto} k$ means that $k$ is obtained by composing $i$ and $j$ via a tagging $\gamma \in \Gamma(i,j)$ as defined in Def.\ref{def:int_compo_via_gate_anti_unif}. We use \emph{mappings} to replace gate with actions or value-passings.  
A mapping is a partial function 
$\lambda : \mathfrak{F}_0 \rightharpoonup (\mathbb{A} \cup \mathbb{VP})$,  
written $\lambda = \{\mathfrak{a}_1 \mapsto u_1, \dots, \mathfrak{a}_n \mapsto u_n\}$, so that $\dom(\lambda) = \{\mathfrak{a}_1, \dots, \mathfrak{a}_n\} \subseteq \mathfrak{F}_0$.  
For $t \in \mathcal{T}(\mathcal{F} \cup \mathfrak{F}_0)$, $t\lambda$ is defined by $t\lambda$ if $t \in \dom(\lambda)$, $t\lambda = t$ if $t \in \mathfrak{F}_0 \setminus \dom(\lambda)$, and $t\lambda = f(t_1\lambda,\dots,t_m\lambda)$ if $t = f(t_1,\dots,t_m)$ with $f \in \mathcal{F}_m$.

\begin{definition}
\label{def:int_compo_via_gate_anti_unif}
For $i,j,k \in \mathbb{I}$ with $\lfs(i)\cap\lfs(j)=\emptyset$ and $\gamma\in\Gamma(i,j)$, we write $(i,j)\overset{\gamma}{\propto}k$ iff:
\begin{itemize}
    \item $\lambda_i = \{\mathfrak{a}\mapsto i|_{p_i} \mid (p_i,p_j,\mathfrak{a})\in\gamma\}$ and
          $\lambda_j = \{\mathfrak{a}\mapsto j|_{p_j} \mid (p_i,p_j,\mathfrak{a})\in\gamma\}$;
    \item $i = s\lambda_i$ and $j = t\lambda_j$ for some $s,t\in\mathcal{T}(\mathcal{F}\cup\mathfrak{F}_0)$;
    \item $r\in\mathcal{T}(\mathcal{F}\cup\mathfrak{F}_0,\V)$ is an $E$-scpg of $s$ and $t$ with witnesses $\sigma_s,\sigma_t$;
    \item $\sigma_r$ satisfies $\dom(\sigma_r)=\dom(\sigma_s)$ and 
          $x\sigma_r = \textsf{seq}(x\sigma_s, x\sigma_t)$ for all $x$;
    \item $\lambda_k$ satisfies $\dom(\lambda_k)=\dom(\lambda_i)$ and 
          $\mathfrak{a}\lambda_k = \textsf{vp}(\mathfrak{a}\lambda_i,\mathfrak{a}\lambda_j)$ for all $\mathfrak{a}$;
    \item $k = r\sigma_r\lambda_k$.
\end{itemize}
\end{definition}

%% file: input/7_2_soundness.tex
To prove the soundness of our composition, we adapt the projection-based approach of Lange~\cite{LangeSynthesis}, where a global type is inferred from local session types. Def.\ref{def-projection} introduces a projection operator $\pi$. In~\cite{MaheBGG25}, we showed that $\pi$ preserves interaction semantics, i.e., $\pi_L(i)$ gives the traces of $i$ restricted to actions on $L$. 

\begin{definition}
\label{def-projection}
The projection $\pi_L$ onto a subset of lifelines $L \subset \mathcal{L}$ is defined by:
\[
\begin{array}{c @{\hspace{.5cm}} c}
\pi_L(l\Delta m) \!=\!
\left\{ 
\begin{array}{ll}
l \Delta m & \textbf{if } l \in L\\
\varnothing & \textbf{if } l \not\in L
\end{array}
\right. 
&
\pi_L(\vp(l_1,m,l_2)) \!=\! 
\left\{ 
\begin{array}{ll}
l_1!m & \textbf{if } l_1 \in L, l_2 \not\in L\\
l_2?m & \textbf{if } l_1 \not\in L, l_2 \in L\\
\vp(l_1,m,l_2) &\textbf{if } l_1 \in L, l_2 \in L\\
\varnothing &\textbf{if } l_1 \not\in L, l_2 \not\in L
\end{array}
\right.
\\ [.75cm]
\pi_L(\textsf{loop}(i)) = \textsf{loop}(\pi_L(i))
&
\forall f \in \{ \seq, \alt, \para\} ~ \pi_L( f(i,j)) = f(\pi_L(i), \pi_L(j))
\end{array}
\]
\end{definition}

With Th.\ref{th:int_compo_via_gate_anti_unif_is_sound_compo}, we show that whenever $(i,j) \overset{\gamma}{\propto} k$, projecting $k$ on the lifelines of $i$ (resp.~$j$) yields an interaction semantically equivalent to $i$ (resp.~$j$). For this, we extend $\theta$ and $\pi$ to $\mathcal{T}(\mathcal{F} \cup \mathfrak{F}_0,\mathcal{V})$ via $\theta(\mathfrak{a}) = \theta(x) = \emptyset$, $\pi_L(\mathfrak{a}) = \mathfrak{a}$, and $\pi_L(x) = x$.

\begin{theorem}
\label{th:int_compo_via_gate_anti_unif_is_sound_compo}
For any $i,j,k \in \inter$ with $\lfs(i) \cap \lfs(j) = \emptyset$ and any $\gamma \in \Gamma(i,j)$, if $(i,j) \overset{\gamma}{\propto} k$ then $k$ is a sound $E$-$\gamma$-composition of $i$ and $j$, meaning:
\[
\pi_{\lfs(i)}(k) =_E i \qquad\text{and}\qquad \pi_{\lfs(j)}(k) =_E j.
\]
\end{theorem}

\begin{proof}
We prove $\pi_{\theta(i)}(k) =_E i$ (the case $\pi_{\theta(j)}(k) =_E j$ is analogous). Let us introduce the terms $s,t,r$, mappings $\mappingNotation_i,\mappingNotation_j,\mappingNotation_k$ and substitutions $\sigma_s,\sigma_t,\sigma_r$ such that the conditions of Def.\ref{def:int_compo_via_gate_anti_unif} are satisfied.

First, note that $\lfs(r) = \emptyset$. Indeed, $r$ is a generalization of $s$ and $t$ (hence $i$ and $j$), and cannot contain non-shared constants. As $\lfs(i) \cap \lfs(j) = \emptyset$, the only shared constants are $\varnothing$ or gate symbols in $\mathfrak{F}_0$, so $\pi_{\theta(i)}(r) = r$.

Then, consider that $\forall x \in \dom(\sigma_r) = \dom(\sigma_s)$, we have $x \sigma_r = \textsf{seq}(x \sigma_s, x \sigma_t)$.
Because $\sigma_s$ is used to build $s$ and then $i$ from $r$, we must have $\lfs(x \sigma_s) \subseteq \lfs(i)$. Similarly, $\lfs(x \sigma_t) \subseteq \lfs(j)$.
Therefore, given $\lfs(i) \cap \lfs(i) = \emptyset$, we have $\pi_{\theta(i)}(x \sigma_r) = \pi_{\theta(i)}(\textsf{seq}(x \sigma_s, x \sigma_t)) =_E x \sigma_s$.
As a result, $\pi_{\theta(i)}(r \sigma_r) =_E \pi_{\theta(i)}(r) \sigma_s = r\sigma_s$.

Consider then that $\forall (\_,\_,\mathfrak{a}) \in \gamma$ we have $\mathfrak{a} \mappingNotation_k = \textsf{vp}(\mathfrak{a} \mappingNotation_i,\mathfrak{a} \mappingNotation_j)$.
Via a similar reasoning, we have $\pi_{\theta(i)}(\mathfrak{a} \mappingNotation_k) = \mathfrak{a} \mappingNotation_i$.
Therefore $\pi_{\theta(i)}(r \sigma_r \mappingNotation_k) = \pi_{\theta(i)}(r \sigma_r) \mappingNotation_i$. 

Finally, as $k = r \sigma_r \mappingNotation_k$, we have
$\pi_{\theta(i)}(k) = \pi_{\theta(i)}(r \sigma_r) \mappingNotation_i =_E r\sigma_s \mappingNotation_i = s \mappingNotation_i = i$. \qed

\end{proof}

%% file: input/7_3_exp.tex
\begin{figure}[t]
    \centering
     \input{table/table_experiments}
  \caption{Experiments on interaction composition using anti-unification}
    \label{fig:tab-benchmark}
\end{figure}

\paragraph{Experiments.} 
We have implemented equational sc-preserving anti-unification in a Rust prototype tool~\cite{generalizer_repo_2025}, which also supports interaction composition and integrates interaction manipulation facilities from an existing Rust-based tool~\cite{hibou_repo_2025}. Experiments were run on an Intel Core i7-13850HX (20-core, 2.1 GHz) with 32 GB RAM, using the \textsf{Fail} (abbreviated \textsf{F}) rule to improve performance, and a timeout threshold of 60s. Fig.\ref{fig:tab-benchmark} summarizes experiments on interactions adapted from the literature~\cite{LangeCFSMtoGlobal,benchmark-professor-online-ref,benchmark-travel-ref,benchmark-two-buyer-protocol-ref}. For each interaction $k$ with lifelines $L$, we randomly generate up to $5$ partitions $(L_1,L_2)$, each of size at least $\lfloor L/2\rfloor$. For each partition: (1) compute projections $(i \!=\! \pi_{L_1}(k),j \!=\! \pi_{L_2}(k))$ and their normal forms $(i_\text{norm}, j_\text{norm})$; (2) apply $7$ random commutativity rewrites $\para(x,y) \approx \para(y,x)$ and $\alt(x,y) \approx \alt(y,x)$ using a rewriting tool~\cite{maude_repo_2025} to produce mutated terms $(i_\text{mut}, j_\text{mut})$; (3) compose $(i_\text{norm}, j_\text{norm})$ and $(i_\text{mut}, j_\text{mut})$, and compare the normal forms of their compositions and that of $k$ to determine success. Each row in Fig.\ref{fig:tab-benchmark} corresponds to an interaction: the second column shows its size, the third the range of gates in local partitions, and the last four columns report average composition durations with and without \textsf{Fail}, for $(i_\text{norm}, j_\text{norm})$ and $(i_\text{mut}, j_\text{mut})$. A global success ($\greencheck$) indicates that the result matches the original interaction. Appendix \ref{appendix_additional_examplesToto}
presents examples of interaction composition with our tool.

%% file: table/table_experiments.tex
\begin{tabularx}{\textwidth}{@{} l c  c c c c c @{}}
        \toprule
        \textbf{Interaction $k$} & 
        \textbf{$\size(k)$} & 
        \textbf{Nb. Gates} & 
        \multicolumn{4}{c}{\textbf{Avg. composition Time (ms)}} \\
        \cmidrule(lr){4-7}
        & & & 
         \multicolumn{2}{c}{$(i_\text{norm},j_\text{norm})$} & 
         \multicolumn{2}{c}{$(i_\text{mut},j_\text{mut})$} \\
         \cmidrule(lr){4-7}
        & & & 
        with \textsf{F} & 
        without \textsf{F} & with \textsf{F} & 
        without \textsf{F} \\
        \midrule
ATM\cite{pomsets_ref}&33&[7, 17]&15.016(\greencheck)&timeout&17.757(\greencheck)&timeout\\
Alt3bit\cite{LangeCFSMtoGlobal}&12&6&4.132(\greencheck)&21.804(\greencheck)&6.01(\greencheck)&32.932(\greencheck)\\
DistVoting\cite{pomsets_ref}&23&8&8.243(\greencheck)&timeout&10.955(\greencheck)&timeout\\
FilterCo\cite{LangeCFSMtoGlobal}&11&5&1.396(\greencheck)&1.565(\greencheck)&2.069(\greencheck)&2.339(\greencheck)\\
Game\cite{LangeCFSMtoGlobal}&16&[5, 6]&2.525(\greencheck)&3.179(\greencheck)&3.163(\greencheck)&4.283(\greencheck)\\
HealthSys\cite{LangeCFSMtoGlobal}&22&[4, 8]&5.978(\greencheck)&18.496(\greencheck)&7.495(\greencheck)&22.786(\greencheck)\\
Logistic\cite{LangeCFSMtoGlobal}&26&[6, 11]&9.22(\greencheck)&timeout&12.457(\greencheck)&timeout\\
ProfOnline\cite{benchmark-professor-online-ref}&68&[14, 27]&59.567(\greencheck)&timeout&75.474(\greencheck)&timeout\\
Sanitary\cite{LangeCFSMtoGlobal}&30&[6, 13]&11.491(\greencheck)&40.242(\greencheck)&14.795(\greencheck)&51.265(\greencheck)\\
TPM\cite{LangeCFSMtoGlobal}&17&7&3.204(\greencheck)&4.911(\greencheck)&4.915(\greencheck)&7.437(\greencheck)\\
Travel\cite{benchmark-travel-ref}&26&[6, 11]&6.411(\greencheck)&timeout&8.203(\greencheck)&timeout\\
TwoBuyers\cite{benchmark-two-buyer-protocol-ref}&22&[5, 11]&3.78(\greencheck)&5.558(\greencheck)&4.852(\greencheck)&7.162(\greencheck)\\
        \bottomrule
    \end{tabularx}

%% file: input/9_conclusion.tex
\section{Conclusion}
We introduced a tool-supported method for composing interactions, represented as terms, from partial views tagged with gates (special constants) via 
sc-preserving anti-unification, which preserves gates while the surrounding structure may be generalized, modulo an equational theory of interactions. Our sc-preserving framework extends and refines anti-unification \cite{ALPUENTE2014,ALPUENTE2022} while upholding expected termination, soundness and completeness properties. Future work includes exploring additional equations, such as idempotence, $\textsf{alt}(x,x) \approx x$, and ad-hoc equations like $\textsf{seq}(\textsf{loop}(x),\textsf{loop}(x)) \approx \textsf{loop}(x)$, to further enhance anti-unification’s ability to compose interactions.

%% file: appendix/Additional_examples.tex
\section{Additional examples of interaction composition} \label{appendix_additional_examplesToto}

In Sec.\ref{sec:application_to_int_compo}, we have presented a composition procedure of interaction models. The following tables show the result of applying our tool to compose selected local interactions. In each case, the local interactions $i$ and $j$ are associated with a tagging that binds compatible communication actions with the same gate. As in Fig.\ref{fig:example_interaction_composition}, the gates are represented with colored squares next to the arrow of their action. The composition then combines the interaction $i$ and $j$ by merging compatible communication actions into value-passings in an interaction $k$, while preserving local behavior.

On Table.\ref{tab:exple1}, the tagging is $\gamma = \{(121,1,\textcolor{red!50!white}{\mathfrak{a}}),(122,221,\textcolor{violet!50!white}{\mathfrak{b}}),(22,222,\textcolor{cyan!50!white}{\mathfrak{c}})\}$ with:

\[
\begin{array}{lllclll}
i|_{121} &=& \hlf{center}!\hms{query} 
&
~~~~~~
& 
j|_{1} &=& \hlf{proxy}?\hms{query}
\\
i|_{122} &=& \hlf{center}?\hms{msg} 
&
~~~~~~
& 
j|_{221} &=& \hlf{proxy}!\hms{msg}
\\
i|_{22} &=& \hlf{center}!\hms{sig}
&
~~~~~~
& 
j|_{222} &=& \hlf{proxy}?\hms{sig}
\end{array}
\]

The gates $\textcolor{red!50!white}{\mathfrak{a}},\textcolor{violet!50!white}{\mathfrak{b}}$ and $\textcolor{cyan!50!white}{\mathfrak{c}}$ are represented with red, purple and blue squares respectively. The composition of $i$ and $j$ gives the interaction k.

\vspace{1cm}

\begin{table}[]
    \centering
    \begin{tabular}{|c|c|c|}
      \hline
        $i$ & $j$ & $k$ \\ \hline 
         \includegraphics[width=0.25\textwidth]{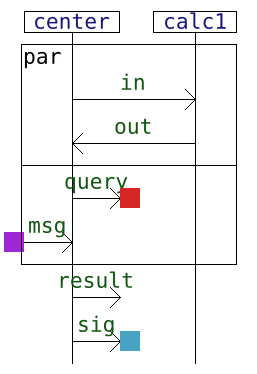}   & \includegraphics[width=0.25\textwidth]{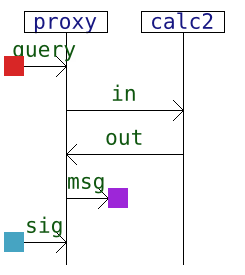} &  \includegraphics[width=0.3\textwidth]{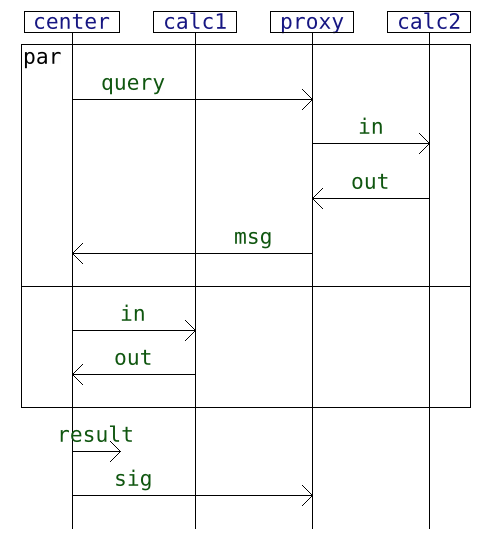}\\ \hline
         \resizebox{.35\textwidth}{!}{\input{figure/example_appendix/i_tree_fig}} & \resizebox{.35\textwidth}{!}{\input{figure/example_appendix/j_tree_fig}} & 
         \resizebox{.35\textwidth}{!}{\input{figure/example_appendix/k_tree_fig}}\\ \hline
    \end{tabular}
    \caption{Illustration of the composition for a simple protocol.}
    \label{tab:exple1}
\end{table}

\begin{table}[]
    \centering
    \begin{tabular}{|c|c|c|}
      \hline
        $i$ & $j$ & $k$ \\ \hline 
         \includegraphics[width=0.32\textwidth]{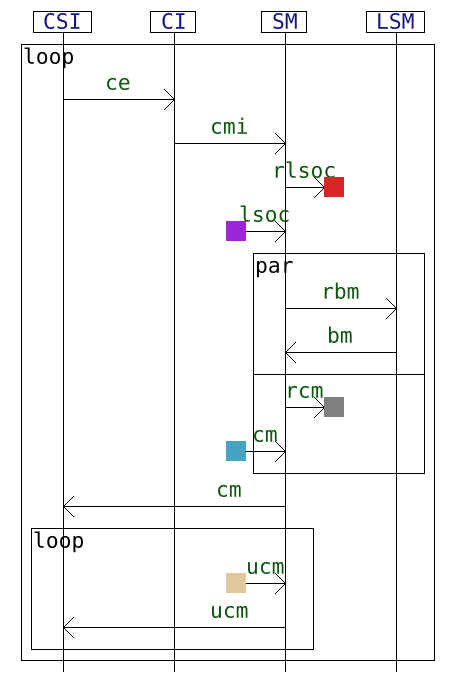}   & \includegraphics[width=0.32\textwidth]{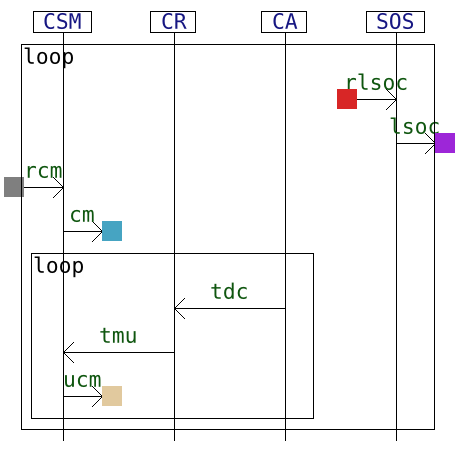} &  \includegraphics[width=0.47\textwidth]{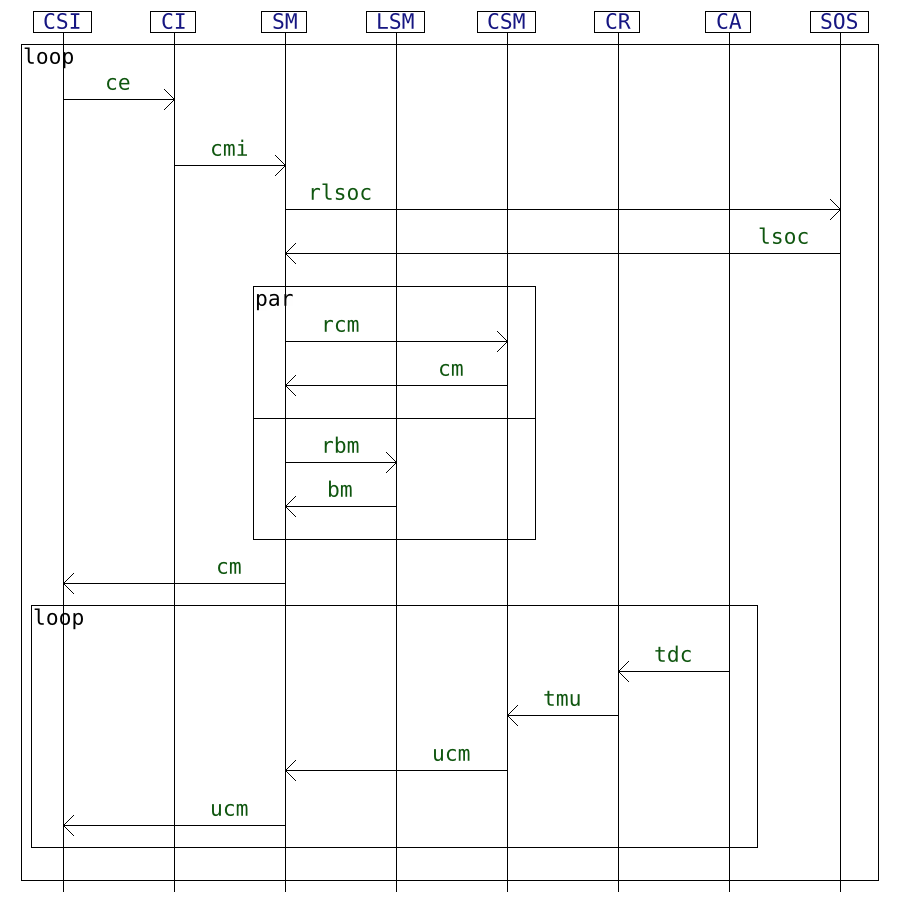}\\ \hline
    \end{tabular}
    \caption{Illustration of the composition for an interaction describing the query of sensor data \cite{Sensor_ref}.}
    \label{tab:expleSensor}
\end{table}

\begin{table}[]
    \centering
    \begin{tabular}{|c|c|c|}
      \hline
        $i$ & $j$ & $k$ \\ \hline 
         \includegraphics[width=0.25\textwidth]{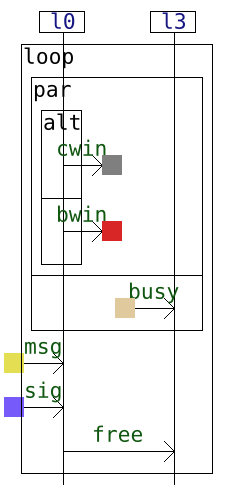}   & \includegraphics[width=0.25\textwidth]{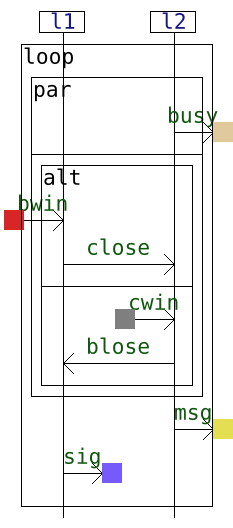} &  \includegraphics[width=0.4\textwidth]{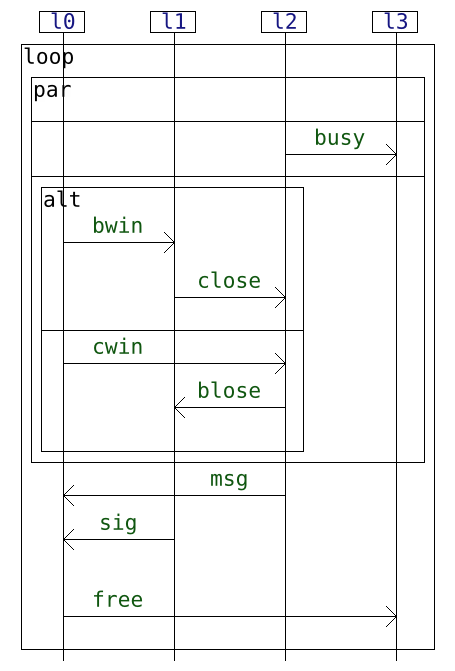}\\ \hline
    \end{tabular}
    \caption{Illustration of the composition for the Game interaction \cite{LangeCFSMtoGlobal}.}
    \label{tab:expleGame}
\end{table}

\begin{table}[]
    \centering
    \begin{tabular}{|c|c|c|}
      \hline
        $i$ & $j$ & $k$ \\ \hline 
         \includegraphics[width=0.25\textwidth]{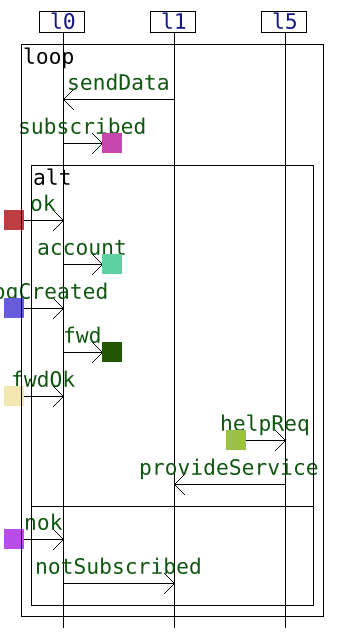}   & \includegraphics[width=0.25\textwidth]{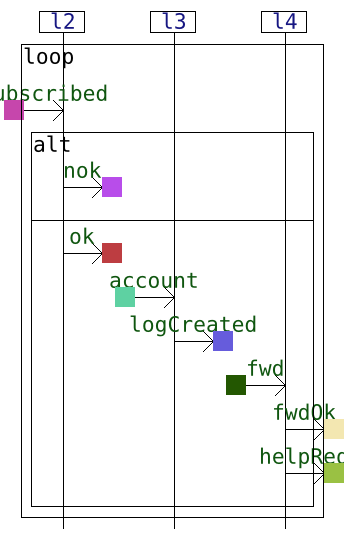} &  \includegraphics[width=0.4\textwidth]{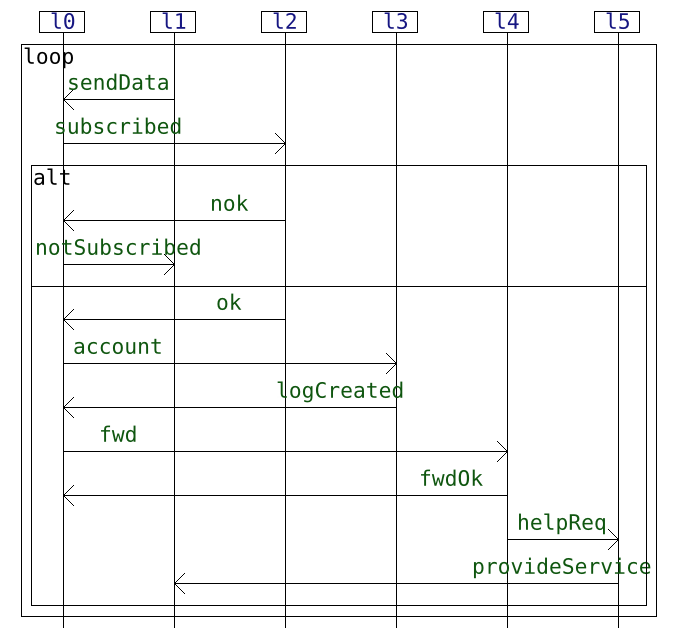}\\ \hline
    \end{tabular}
    \caption{Illustration of the composition for the Health System interaction \cite{LangeCFSMtoGlobal}.}
    \label{tab:expleHS}
\end{table}


\begin{table}[]
    \centering
    \begin{tabular}{|c|c|c|}
      \hline
        $i$ & $j$ & $k$ \\ \hline 
         \includegraphics[width=0.15\textwidth]{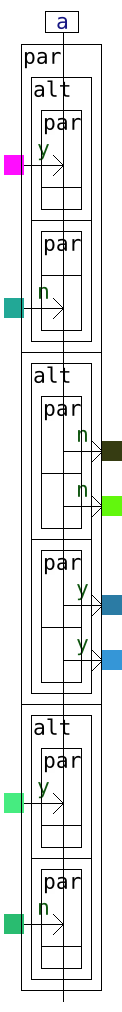}   & \includegraphics[width=0.25\textwidth]{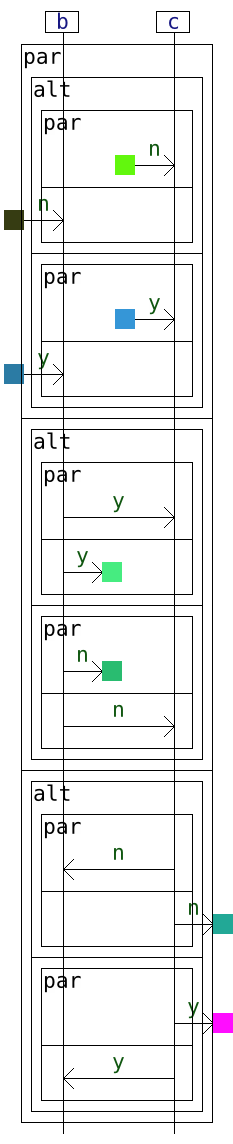} &  \includegraphics[width=0.4\textwidth]{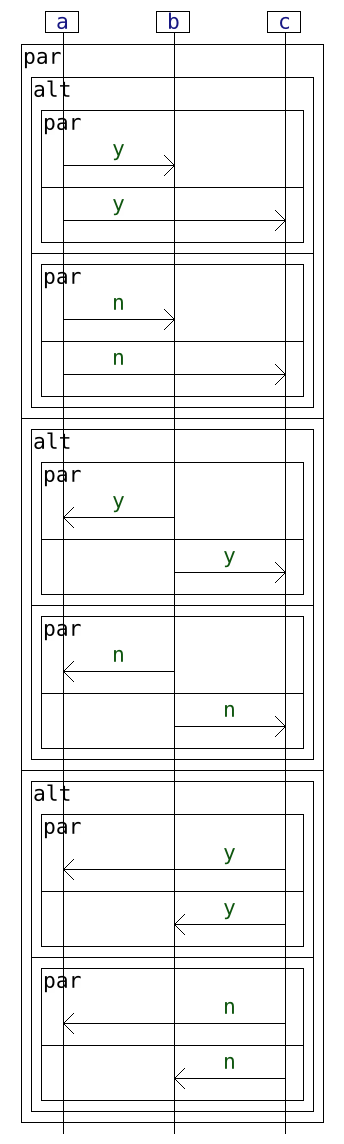}\\ \hline
    \end{tabular}
    \caption{Illustration of the composition for a Distributed Voting protocol \cite{pomsets_ref}.}
    \label{tab:expleDV}
\end{table}

\begin{table}[]
    \centering
    \begin{tabular}{|c|c|c|}
      \hline
        $i$ & $j$ & $k$ \\ \hline 
         \includegraphics[width=0.26\textwidth]{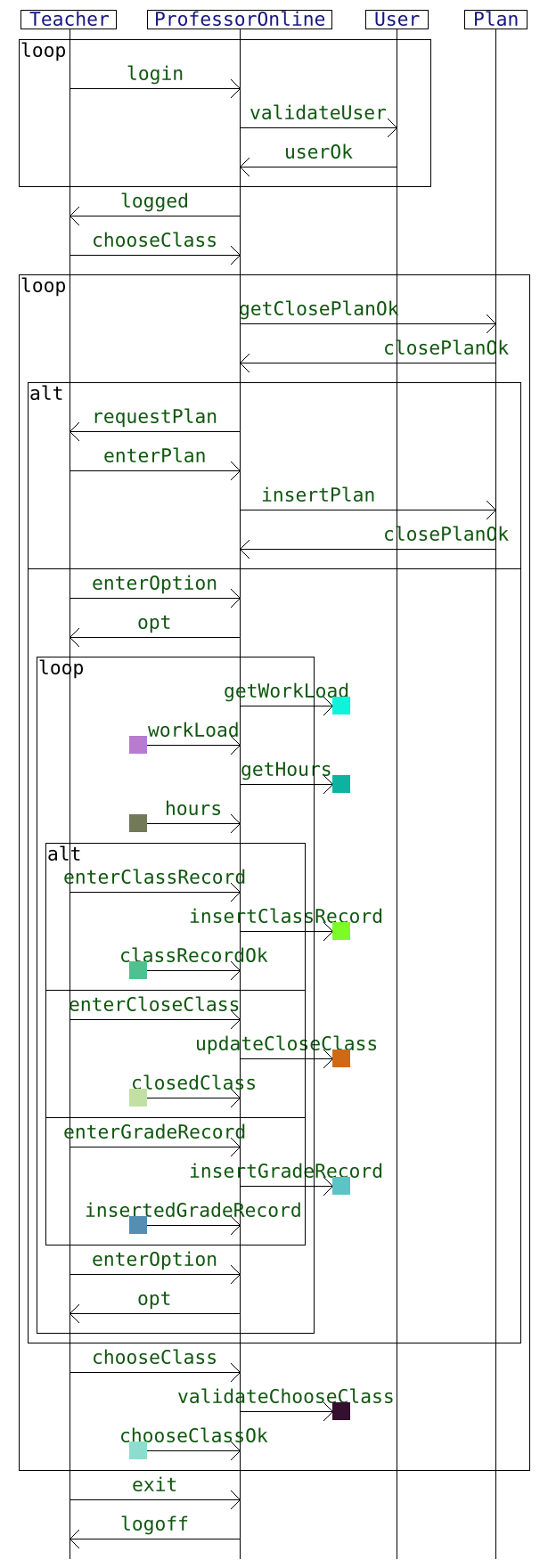}   & \includegraphics[width=0.3\textwidth]{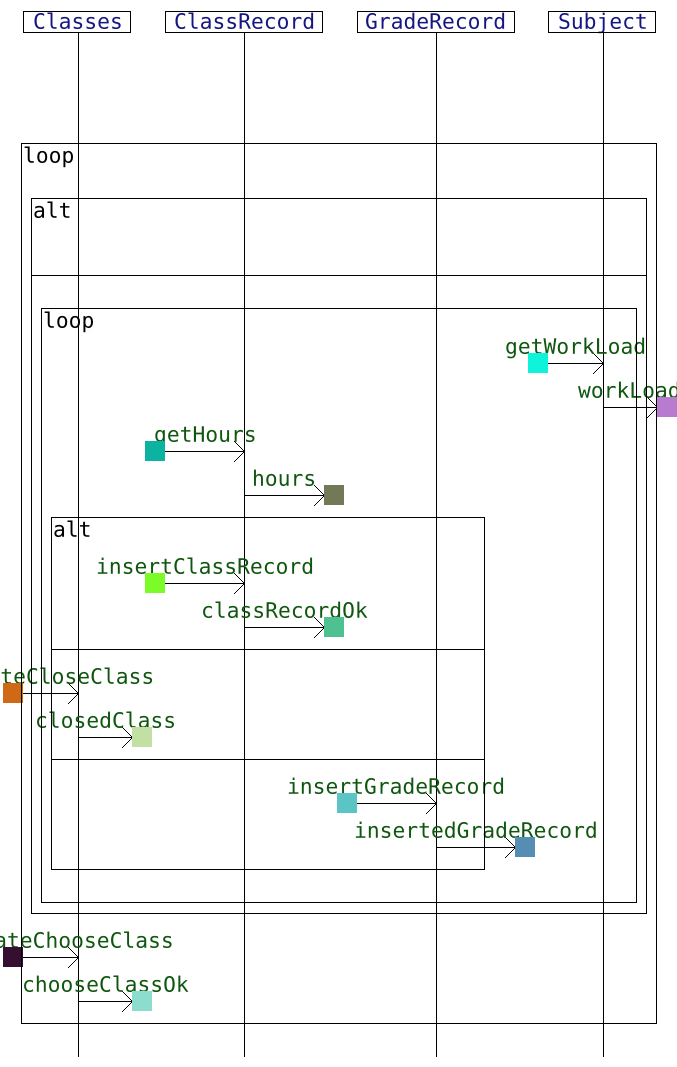} &  \includegraphics[width=0.44\textwidth]{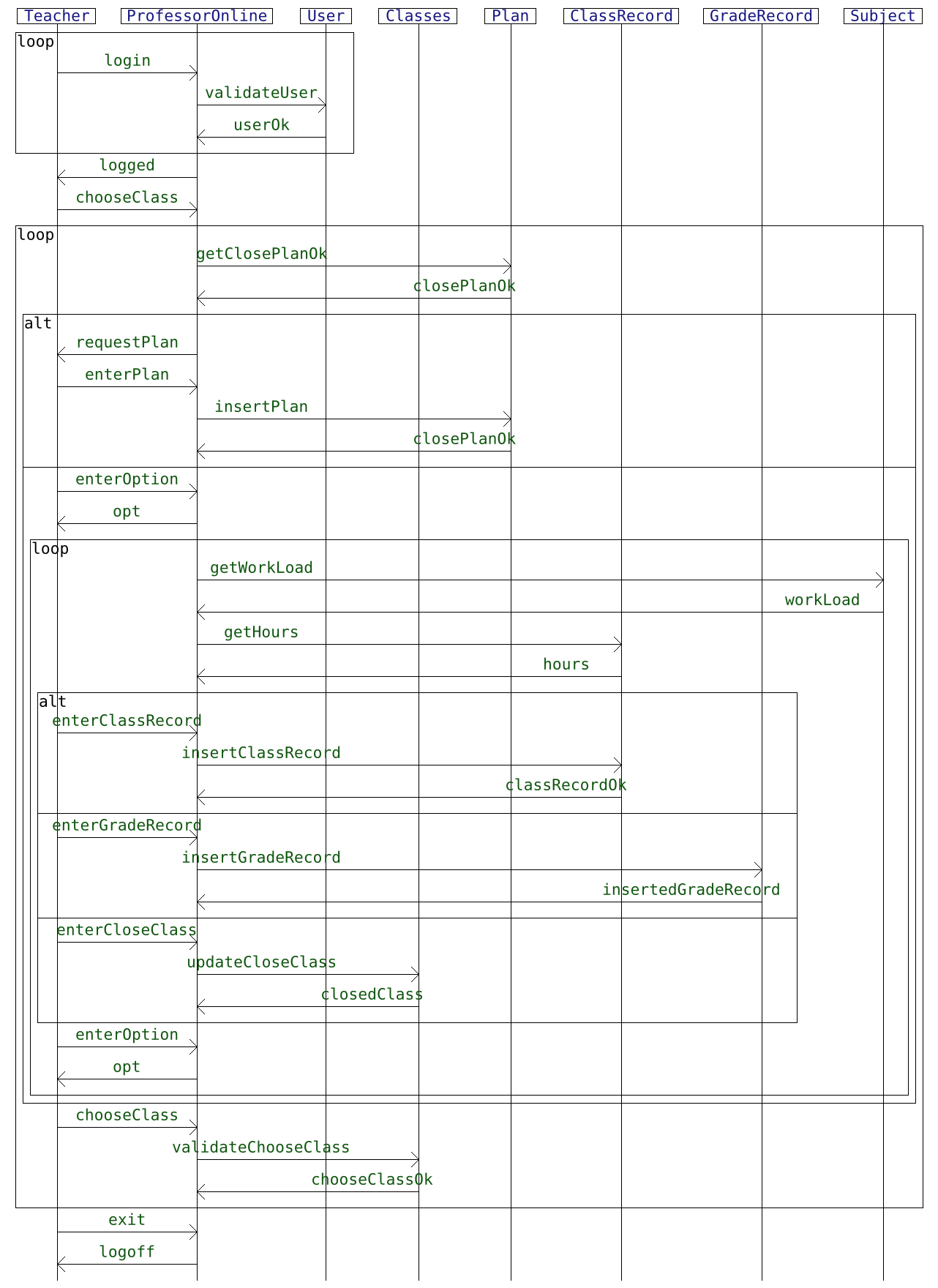}\\ \hline
    \end{tabular}
    \caption{Illustration of the composition for the Professor online interaction \cite{benchmark-professor-online-ref}}.
    \label{tab:expleProfOnline}
\end{table}

%% file: figure/example_appendix/i_tree_fig.tex
\begin{tikzpicture}[scale=0.7, transform shape,]
    \node (o) {\textsf{seq}} [sibling distance=2cm, level distance=0.5cm]
    child { node (o1) {\textsf{par}} [sibling distance=0.8cm, level distance=0.5cm]
        child { 
            node[xshift=-20mm,yshift=-30mm] (o11) {\textsf{seq}}
            child{
                node[yshift=-2mm] (o111) {\vp(\hlf{center}!\hms{in},\hlf{calc1}?\hms{in})}            
            }
            child{
                node[xshift=30mm,yshift=-2mm] (o112) {\vp(\hlf{calc1}!\hms{out},\hlf{center}?\hms{out})}
            }
        }
        child{
             node[yshift=-1cm] (o12) {\textsf{seq}} [sibling distance=1.8cm, level distance=1cm]
             child{
                node (o121) {$\hlf{center}!\hms{query}$}
                 node[below=0.25,draw,fill=red!50!white] (pos1) {121}
             }
             child{
                node (o122) {$\hlf{center}?\hms{msg}$}
                 node[below=0.25,draw,fill=violet!50!white] (pos1) {122}
             }
        }
    }
    child { node (o2) {\textsf{seq}} [sibling distance=1.3cm, level distance=0.8cm]
        child{
            node (o21) {$\hlf{center}!\hms{result}$}
        }
        child{
            node[yshift=-4mm] (o22) {$\hlf{center}!\hms{sig}$}
             node[below=0.55,draw,fill=cyan!50!white] (pos1) {22}
        }
    };
\end{tikzpicture}

%% file: figure/example_appendix/j_tree_fig.tex
\begin{tikzpicture}[scale=0.7, transform shape,]
    \node (o) { \textsf{seq} } [sibling distance=1.5cm, level distance=0.6cm]
    child{
        node (o1) {$\hlf{proxy}?\hms{query}$}
        node[below=0.25,draw,fill=red!50!white] (pos1) {1}
    }
    child{
        node (o2) { \textsf{seq}} 
        child{
            node[xshift=-8mm,yshift=-13mm] (o21) {\textsf{seq}}
            child{
                node[xshift=-3mm,yshift=-4mm] (o211) {\vp(\hlf{proxy}!\hms{in},\hlf{calc2}?\hms{in})}
            }
            child{
                node[xshift=15mm,yshift=-4mm] (o212) {\vp(\hlf{calc2}!\hms{out},\hlf{proxy}?\hms{out})}
            }
        }
        child{
            node(o22) {\textsf{seq}}
            child{
                node (o221) {$\hlf{proxy}!\hms{msg}$}
                 node[below=0.25,draw,fill=violet!50!white] (pos1) {221}
            }
            child{
                node (o222) {$\hlf{proxy}?\hms{sig}$}
                 node[below=0.25,draw,fill=cyan!50!white] (pos1) {222}
            }
        }
    };
\end{tikzpicture}

%% file: figure/example_appendix/k_tree_fig.tex
\begin{tikzpicture}[scale=0.6, transform shape,]
    \node (o) {\textsf{seq}} [sibling distance=2cm, level distance=0.5cm]
    child { node (o1) {\textsf{par}} [sibling distance=0.8cm, level distance=0.5cm]
                child{
                     node[yshift=-1cm] (o11) {\textsf{seq}} [sibling distance=1.8cm, level distance=1cm]
                     child{
                        node[xshift=-6mm,draw,fill=red!50!white] (o111) {\vp(\hlf{center}!\hms{query},\hlf{proxy}?\hms{query})}
                     }
                     child{
                        node[xshift=12mm,yshift=-1cm] (o112) {\textsf{seq}}
                        child { 
                            node (o1121) { \textsf{seq}} 
                                    child{
                                        node[xshift=-3mm,yshift=-4mm] (o211) {\vp(\hlf{proxy}!\hms{in},\hlf{calc2}?\hms{in})}
                                    }
                                    child{
                                        node[xshift=12mm,yshift=-4mm] (o212) {\vp(\hlf{calc2}!\hms{out},\hlf{proxy}?\hms{out})}
                                    }
                            }
                        child{ 
                            node[xshift=10mm,draw,fill=violet!50!white] (o1122) {\vp(\hlf{proxy}!\hms{msg},\hlf{center}?\hms{msg})}
                        }
                }
            }
        child { 
            node[xshift=26mm,yshift=-16mm] (o12) {\textsf{seq}}
            child{
                node[yshift=-2mm] (o1121) {\vp(\hlf{center}!\hms{in},\hlf{calc1}?\hms{in})}            
            }
            child{
                node[xshift=25mm,yshift=-2mm] (o122) {\vp(\hlf{calc1}!\hms{out},\hlf{center}?\hms{out})}
            }
        }
    }
    child { node[xshift=10mm] (o2) {\textsf{seq}} [sibling distance=1.3cm, level distance=0.8cm]
        child{
            node (o21) {$\hlf{center}!\hms{result}$}
        }
        child{
            node[xshift=6mm,yshift=-4.2mm,draw,fill=cyan!50!white] (o22) {\vp(\hlf{center}!\hms{sig},\hlf{proxy}?\hms{sig})}
        }
    };
\end{tikzpicture}